\documentclass{article}

\usepackage[affil-it]{authblk}
\usepackage[usenames,dvipsnames]{xcolor}
\usepackage{amsfonts}
\usepackage{amsmath,amsthm,amssymb,dsfont}
\usepackage{enumerate}
\usepackage[english]{babel}
\usepackage{graphicx}	
\usepackage[caption=false]{subfig}
\usepackage[margin=3cm]{geometry}
\usepackage{url}
\usepackage{todonotes}
\usepackage{bbm}

\usepackage{tikz}
\usetikzlibrary{chains}
\usetikzlibrary{fit}

\usepackage{makecell}

\usepackage{epsfig}
\usetikzlibrary{shapes.symbols,patterns} 
\usepackage{pgfplots}

\usepackage{hyperref}[breaklinks]
\hypersetup{colorlinks=true,citecolor=blue,linkcolor=blue,filecolor=blue,urlcolor=blue}

\usepackage{nicefrac}
\usepackage{mathtools}

\usepackage{algorithm}
\usepackage{algorithmic}
\usepackage{wasysym}

\usepackage{booktabs}

\theoremstyle{plain}
\newtheorem{theorem}{Theorem}[section]
\newtheorem{lemma}[theorem]{Lemma}

\newtheorem{fact}[theorem]{Fact}
\newtheorem{corollary}[theorem]{Corollary}
\newtheorem{proposition}[theorem]{Proposition}

\theoremstyle{definition}

\newtheorem{remark}[theorem]{Remark}

\newtheorem{assumption}[theorem]{Assumption}
\newtheorem{problem}[theorem]{Problem}

\newcommand*{\PP}{\mathbb{P}}
\newcommand*{\E}{\mathbb{E}}

\newcommand*{\cA}{\mathcal{A}}

\newcommand*{\cL}{\mathcal{L}}

\newcommand*{\cM}{\mathcal{M}}

\newcommand*{\cQ}{\mathcal{Q}}
\newcommand*{\cR}{\mathcal{R}}
\newcommand*{\cS}{\mathcal{S}}

\newcommand*{\cU}{\mathcal{U}}

\newcommand*{\cW}{\mathcal{W}}
\newcommand*{\cX}{\mathcal{X}}
\newcommand*{\cY}{\mathcal{Y}}
\newcommand*{\cZ}{\mathcal{Z}}

\newcommand*{\X}{\mathbb{X}}
\newcommand*{\U}{\mathbb{U}}
\newcommand*{\V}{\mathbb{V}}
\newcommand*{\Y}{\mathbb{Y}}
\newcommand*{\M}{\mathbb{M}}

\newcommand*{\N}{\mathbb{N}}
\newcommand*{\Z}{\mathbb{Z}}
\newcommand*{\R}{\mathbb{R}}
\newcommand*{\Q}{\mathbb{Q}}
\newcommand*{\C}{\mathbb{C}}

\newcommand*{\eps}{\varepsilon}

\newcommand*{\poly}{\mathrm{poly}}
\newcommand*{\polylog}{\mathrm{polylog}}

\newcommand*{\ket}[1]{| #1 \rangle}
\newcommand*{\bra}[1]{\langle #1 |}

\newcommand*{\conv}{\mathrm{conv}}

\newcommand*{\distH}{\mathrm{d_H}}

\DeclareSymbolFont{mymathoperators}{OT1}{phv}{m}{n}
\DeclareMathSymbol{\protoast}{\mathbin}{mymathoperators}{"2A}
\renewcommand*{\ast}{\mathbin{\raisebox{-0.7ex}{\ensuremath{\protoast}}}}

\newcommand*{\ci}{\mathrm{i}} 
\newcommand*{\dr}{\mathrm{r}} 
\newcommand*{\di}{\mathrm{d}} 

\newcommand{\norm}[1]{\left\lVert#1\right\rVert}

\usetikzlibrary{calc}
\usetikzlibrary{shapes.geometric}
\tikzset{
    sunflames/.style={
        line width=0.1pt,
        draw=black,
        fill=yellow,
        regular polygon, 
        regular polygon sides=3,
        inner sep=0.035cm
    },
    sunbody/.style={
        line width=0.1pt,
        draw=black,
        fill=yellow,
        circle,
        minimum size=0.5cm
    }
 }

\newcommand{\inprod}[2]{\ensuremath{\left\langle{#1}\vphantom{\big|},\vphantom{\big|}{#2}\right\rangle}}

\allowdisplaybreaks

\title{Quantum speedups for convex dynamic programming}

\author{\normalsize David Sutter$^{1}$, Giacomo Nannicini$^{2}$, Tobias Sutter$^{3}$, and Stefan Woerner$^{1}$}
\affil{\small $^{1}$IBM Quantum, IBM Research -- Zurich\\
\small $^{2}$IBM Quantum, IBM T.J.~Watson Research Center\\
\small $^{3}$Risk Analytics and Optimization Chair, EPFL
}
\date{}

\begin{document}

\maketitle

\begin{abstract}
We present a quantum algorithm to solve dynamic programming problems with convex value functions. For linear discrete-time systems with a $d$-dimensional state space of size $N$, the proposed algorithm outputs a quantum-mechanical representation of the value function in time $O(T \gamma^{dT}\polylog(N,(T/\eps)^{d}))$, where $\eps$ is the accuracy of the solution, $T$ is the time horizon, and $\gamma$ is a problem-specific parameter depending on the condition numbers of the cost functions. This allows us to evaluate the value function at any fixed state in time $O(T \gamma^{dT}\sqrt{N}\,\polylog(N,(T/\eps)^{d}))$, and the corresponding optimal action can be recovered by solving a convex program. The class of optimization problems to which our algorithm can be applied includes provably hard stochastic dynamic programs. Finally, we show that the algorithm obtains a quadratic speedup (up to polylogarithmic factors) compared to the classical Bellman approach on some dynamic programs with continuous state space that have $\gamma=1$.
\end{abstract}

\paragraph{Update:}
We recently discovered an error in the correctness proof of the quantum Legendre-Fenchel transform~\cite[Algorithm~5]{QLFT20}, and we have not found a way to fix it yet. Therefore we are unsure if there is a quantum speedup for the convex dynamic programming setting discussed here. We believe that our hardness results and lower bounds are correct. We will update this manuscript once we have fully resolved this issue.

\section{Introduction}
Quantum computers utilize intrinsic properties of quantum mechanics to perform certain computations with a provable speedup compared to any known classical methods.
Recent technological progress in building quantum hardware strengthens the potential impact of quantum computing, and motivates the development of quantum algorithms that can outperform their classical counterparts. 

In many real-world decision-making problems, decisions are made in stages. Each decision affects the position of the decision-maker in subsequent time steps; the goal is to minimize a measure of overall cost. Often, the problem is further complicated by non-deterministic elements (i.e., costs or state transitions).
The workhorse for this general class of sequential decision-making problems is \emph{dynamic programming} (DP), which dates back to the seminal work of Richard Bellman~\cite{bellman2003dynamic}, and is an established discipline in computer science, operations research, and engineering; for comprehensive references on DP, see~\cite{bertsekas1995dynamic, ref:Bertsekas-12}. In machine learning, DP techniques are increasingly popular under the name of \emph{reinforcement learning}.
The fundamental limitation of DP, captured by Bellman's term ``curse of dimensionality'', is the exponential growth in the dimension of the state and action space of the running time of the standard recursive solution approach, coupled with the fact that all different running time parameters are multiplied together. Thus, even if the standard solution approach is efficient in practice for small problems, it can quickly grow intractable. To solve DP problems, a possibility is to rely on the approximation techniques at the core of \emph{approximate dynamic programming} (ADP) \cite{ref:BerTsi-96,powell2007approximate}.
While ADP cannot resolve the curse of dimensionality in general, it can sometimes lead to a tractable solution approach under some, typically restrictive, conditions, see, e.g., \cite{goldberg2018beating,halman2014fully,halman2019toward}.

It is well known that quantum algorithms can provide a speedup over the best known classical algorithm for several tasks. For example, quadratic speedups over any classical algorithms can be obtained for unstructured search \cite{grover96} and for many computational problems related to Markov chains \cite{szegedy2004quantum}. Because DP is NP-hard, we do not expect more than a quadratic speedup from quantum algorithm \cite{bennett97}. However, attaining such a speedup has escaped so far. 

In this paper we present a quantum algorithm that achieves a quadratic speedup (up to polylogarithmic factors) for certain classes of dynamic programs; specifically, our algorithm only works for problems with a convex structure. While our framework is far from comprising the full generality of DP, it contains several problems that are relevant in practice, and we show that it contains \#P-hard problems. Our paper thus gives a positive, albeit partial, answer to the question of finding a quantum algorithm for DP that is quadratically faster than classical methods.

\subsection{Dynamic programming} \label{sec_DP}
We briefly introduce the DP formalism, distinguishing between a deterministic and a stochastic setting. We refer the reader to the textbooks~\cite{bertsekas1995dynamic,ref:Bertsekas-12} for more details.
\paragraph{Deterministic setting.}
Consider a model comprising a \emph{state space} $\X^d\subseteq \R^d$ and an \emph{action space} $\U^c \subseteq \R^c$ with a deterministic discrete-time system evolving according to the equation
\begin{align*}
x_{t+1} = f(x_t,u_t) \quad \textnormal{for} \quad  t=0,1,\ldots,T-1 \,,
\end{align*}
where $f:\X^d \times \U^c \to \X^d$ describes the system dynamics, $x_0 \in \X^d$ is a given initial state, and $T \in \N$ denotes the time horizon. 
The goal is to find a sequence of optimal actions $\{u^\star_t\}_{t=0}^{T-1}\subseteq \U^c$ via optimal \emph{policies}, i.e., functions $\pi^\star_t:\X^d\to \U^c$ mapping states into actions, with $u^\star_t=\pi^\star_t(x)$, and minimizing the overall cost
\begin{equation*}
\sum_{t=0}^{T-1} g(x_t,u_t) + g_T(x_T)\, ,
\end{equation*}
where $g:\X^d\times \U^c \to \R$ denotes the \emph{running cost} and $g_T:\X^d\to\R$ the \emph{terminal cost}. The optimal total cost achieved from a given stage to the terminal stage is given by the \emph{value function} of the problem, which for the first stage is defined as
\begin{equation}\label{eq:value:fct:def}
J_0(x) := \min_{\pi_0,\ldots,\pi_{T-1}} \left \lbrace \sum_{t=0}^{T-1} g\big(x_t,\pi_t(x_t)\big) + g_T(x_T)\right \rbrace\, .
\end{equation}
The corresponding minimizer, denoted by $\pi^\star_0,\ldots, \pi^\star_{T-1}$ is the \emph{optimal policy}.
The theory of DP states that the optimization problem \eqref{eq:value:fct:def} can be solved by recursively applying the \emph{dynamic programming operator}, defined as
\begin{align} \label{eq_DP}
\mathrm{DP}[J](x) := \min_{u \in \U^c} \{ g(x,u) + J \big( f(x,u)\big) \} \quad \forall x \in \X^d \, ,
\end{align}
where $J:\X^d \to \R$. More precisely, the DP principle \cite[Proposition~1.3.1]{bertsekas1995dynamic} states that we can compute the optimal value function,  defined in \eqref{eq:value:fct:def},  through a sequence of functions recursively defined as
\begin{align} \label{eq:DP:principle}
J_t(x):=\mathrm{DP}[J_{t+1}](x) \qquad \text{for }t=T-1,T-2,\ldots,0\quad \text{and} \quad  \forall x \in \X^d \, ,
\end{align}
where $J_T = g_T$ is the terminal condition.
In technical terms this means
\begin{align}
J_0(x)&=\mathrm{DP}\circ \ldots \circ \mathrm{DP}[J_T](x)  =\mathrm{DP}^T[J_T](x) \qquad \textnormal{and} \label{eq_recursion} \\
\pi_0^\star(x)&\in\arg\min_{u \in \U^d}\big\{g(x,u)+J_0\big(f(x,u)\big) \big\} \, ,\label{eq_recursion_policy}
\end{align}
for all $x\in \X^d$. Hence, we can solve the original problem by applying $T$ times the DP operator in~\eqref{eq_DP}. 
This is a considerable simplification of~\eqref{eq:value:fct:def}: rather than optimizing over policies (i.e.~functions), one needs to solve a sequence of problems of the form~\eqref{eq_DP}, which optimize over $\U^c$.

Despite the simplification via the DP principle, solving Problem~\eqref{eq:value:fct:def} via~\eqref{eq_recursion} is difficult for at least two reasons: (i) evaluating the DP operator~\eqref{eq_DP} requires to solve an optimization problem for each $x \in \X^d$, where $\X^d$ is a possibly infinite set, and (ii) even for a fixed $x \in \X^d$, the optimization over $u \in \U^c$ in~\eqref{eq_DP} may be nontrivial to carry out. We make several assumptions that help us to simplify~\eqref{eq_DP}:
\begin{assumption}[Modelling assumptions]\label{ass:modelling}
We consider states $x=(y,z)$ and decisions $u=(v,w)$ with $y\in \Y^{d_{\dr}} \subset \R^{d_{\dr}}$, $z\in \cZ^{d_{\ci}}_{N_{\ci}} \subset \Z^{d_{\ci}}$, $|\cZ^{d_{\ci}}_{N_{\ci}}|=N_{\ci}$, $v\in \V^{c_{\dr}} \subset \R^{c_{\dr}}$, $w \in \cW^{c_{\ci}}_{M_{\ci}} \subset \Z^{c_{\ci}}$, $|\cW^{c_{\ci}}_{M_{\ci}}|=M_{\ci}$, and define \smash{$\X^d:=\Y^{d_{\dr}}\times \cZ^{d_{\ci}}_{N_{\ci}}$, $\U^c:=\V^{c_{\dr}}\times \cW^{c_{\ci}}_{M_{\ci}}$}, $d:=d_{\dr}+ d_{\ci}$, and $c:=c_{\dr}+ c_{\ci}$. 
Furthermore, the following assumptions hold:
\begin{enumerate}[(i)]
\item\label{it_dynamics} \emph{Linear dynamics}: we have
\begin{align*}
    f(x,u)=\underbrace{\begin{pmatrix}A & 0\\0 & D\end{pmatrix}}_{A'}\underbrace{\begin{pmatrix}y \\z\end{pmatrix}}_{x} + \underbrace{\begin{pmatrix}B &C \\0 &E\end{pmatrix}}_{B'}\underbrace{\begin{pmatrix}v \\w\end{pmatrix}}_{u} \, , 
\end{align*}
for some fixed $A\in \R^{d_{\dr}\times d_{\dr}}$, $B \in \R^{d_{\dr} \times c_{\dr}}$, $C \in \R^{d_{\dr} \times c_{\ci}}$, $D \in \Z^{d_{\ci} \times d_{\ci}}$, and $E \in \Z^{d_{\ci} \times c_{\ci}}$.
\item \label{item:ass:separable} \emph{Separable cost function}: we have $g(x,u)=g_{\mathrm{x}}(x) + g_{\mathrm{u}}(u) = g_{\mathrm{x}}(y,z) + g_{\mathrm{u}}(v,w)$ for some $g_{\mathrm{x}}:\Y^{d_{\dr}}\times \cZ^{d_{\ci}}_{N_{\ci}} \to \R$ and $g_{\mathrm{u}}:\V^{c_{\dr}} \times \cW^{c_{\ci}}_{M_{\ci}} \to \R$.\footnote{Without loss of generality we may assume that \smash{$\Y^{d_{\dr}}=[0,\tau_{\dr}]^{d_{\dr}}$} for \smash{$\tau_{\dr}<\infty$}, \smash{$\max\{ \norm{z}_\infty : z \in \cZ^{d_{\ci}}_{N_{\ci}} \} \leq \tau_{\ci}<\infty$}, \smash{$\V^{c_{\dr}}=[0,\eta_{\dr}]^{c_{\dr}}$} for $\eta_{\dr}<\infty$, and \smash{$\max\{ \norm{w}_{\infty} : w \in \cW^{c_{\ci}}_{M_{\ci}} \} \leq \eta_{\ci}<\infty$}.}
\item \label{item:ass:convexity} \emph{Convexity}: for every $z\in\cZ^{d_{\ci}}_{N_{\ci}}$ the functions $y\mapsto J_T(y,z)$ and $y\mapsto g_{\mathrm{x}}(y,z)$ are convex. For every \smash{$y\in\Y^{d_{\dr}}$} the functions $z\mapsto J_T(y,z)$ and $z\mapsto g_{\mathrm{x}}(y,z)$ are convex extensible. Similarly, for every \smash{$w\in\cW^{c_{\ci}}_{M_{\ci}}$} the function $v\mapsto g_{\mathrm{u}}(v,w)$ is convex and for every $v\in\V^{c_{\dr}}$ the function $w\mapsto g_{\mathrm{u}}(v,w)$ is convex extensible. The convex extensions of $J_T$, $g_{\mathrm{x}}$ and $g_{\mathrm{u}}$ are denoted by $\bar J_T$, $\bar g_{\mathrm{x}}$ and $\bar g_{\mathrm{u}}$, respectively, and are defined in~\eqref{eq:convex:extension}. The sets $\Y^{d_{\dr}}$ and $\V^{c_{\dr}}$ are compact and convex.
\item \label{item:ass:Lipschtiz} \emph{Lipschitz continuity}: the functions $\bar J_T$, and $\bar g_{\mathrm{x}}$ are Lipschitz continuous with constants $L_{\bar J_T}$, and $L_{\bar g_{\mathrm{x}}}$, respectively.
\item \label{item:ass:ConditionNumber} \emph{Finite condition number}: the functions $\bar J_T$, $\bar g_{\mathrm{x}}$, and $\bar g_{\mathrm{u}}$ are differentiable and have a finite condition number denoted by $\kappa_{\bar J_T}$, $\kappa_{\bar g_{\mathrm{x}}}$ and $\kappa_{\bar g_{\mathrm{u}}}$.\footnote{The condition number of a function $f$ is defined as \smash{$\kappa_f:=L'_f/\mu_f$}, where \smash{$L'_f$} is the Lipschitz constant of the gradient of $f$ and \smash{$\mu_f$} denotes the strong convexity parameter.} 
\end{enumerate}
\end{assumption}
We remark that our algorithm allows the matrices $A$, $B$, $C$, $D$, and $E$, as well as the cost functions $g_{\mathrm{x}}$, $g_\mathrm{u}$, to depend on the stage index $t$. However, we assume that they are stage-independent to simplify the exposition. Notice that the state and action spaces are so-called mixed-integer vectors, with $d_\dr, c_\dr$ real components and $d_\ci, c_\ci$ integer components, respectively. Two important cases, that warrant special treatment in some parts of this paper, are the \emph{purely discrete} ($d_\dr = c_\dr = 0$) and \emph{purely continuous} ($d_\ci = c_\ci = 0$) case. Although the above assumptions are restrictive, they are satisfied by many interesting problems, as will be discussed subsequently in the paper. The main motivation for Assumption~\ref{ass:modelling} is that we want to use the quantum Legendre-Fenchel transform (QLFT) \cite{QLFT20} as a subroutine, which inherently requires convex functions. 
Assumption~\ref{ass:modelling}~\eqref{item:ass:ConditionNumber} can be weakened to non-differentiable functions by requiring that $\bar J_T$, $\bar g_{\mathrm{x}}$, and $\bar g_{\mathrm{u}}$ have a finite $W$-parameter as introduced in~\cite[Section~3.1]{QLFT20} and defined in~\eqref{eq_W}. This can be relevant when considering, for example, piecewise-linear cost functions.
Under Assumption~\ref{ass:modelling} the DP operator simplifies to
\begin{align} \label{eq_DP_tilde}
\mathrm{DP}[J](x) 
= g_{\mathrm{x}}(x) + \min_{u \in \U^c} \{ g_{\mathrm{u}}(u) + J( A'x+B'u) \}
=:g_{\mathrm{x}}(x) + \mathrm{DP_{shift}}[J](x) \, ,
\end{align}
where $\mathrm{DP_{shift}}[J]$ denotes the \emph{shifted DP operator}.

The dynamical system defined in Assumption~\ref{ass:modelling} contains continuous as well as discrete variables. To apply the recursion \eqref{eq:DP:principle}, a possible approach is to discretize $\Y^{d_{\dr}}$, by replacing it with \smash{$\cY_{N_{\dr}}^{d_{\dr}}=\{y_0,\ldots,y_{N_{\dr}-1}\} \subseteq \Y$}. This leads to some discretization error, quantified subsequently in the paper.

To simplify notation we denote \smash{$\cX^d_N:=\cY_{N_{\dr}}^{d_{\dr}}\times \cZ^{d_{\ci}}_{N_{\ci}}$}, where $N:=N_{\dr} N_{\ci}$ and $d:=d_{\dr} + d_{\ci}$. 
A function $J':\cX^d_N \to\R$ is called \emph{convex extensible} if there exists a convex function $\bar {J'}:\conv(\cX_{N}^{d})\to\R$ such that $\bar {J'}(x) = J'(x)$ for all $x\in \cX_{N}^{d}$, where $\conv(\cX_{N}^{d})$ denotes the convex hull of $\cX_{N}^{d}$.
We discuss sufficient conditions for a function defined on a discrete set to be convex extensible in Section~\ref{sec_classical}. We next state our last assumption on the model.
\begin{assumption}[Feasibility]\label{ass:feasibility}\
\begin{enumerate}[(i)]
\item\label{it_ass1} Original problem: For every value function $J: \X^d \to \R$ and for every $x \in \X^d$ there exists $u_x^\star \in \arg\min_{u\in\U^{c}} \{g_{\mathrm{u}}(u) + J(A'x + B'u)\} $ such that $A'x+B'u_x^\star \in \X^d$.
\item\label{it_ass2} Conjugate problem: For every value function $J': \cX_N^d \to \R$ and for every $x \in \cX_N^d$ there exists $u_x^\star \in \arg\min_{u\in\U^{c}} \{g_{\mathrm{u}}(u) + \bar {J'}(A'x + B'u)\} $ such that $A'x+B'u_x^\star \in \X^d$.
\end{enumerate}
\end{assumption}
Assumption~\eqref{it_ass1} ensures that the original DP problem is feasible. Assumption~\eqref{it_ass2} extends this feasibility to the discretized problem, which is a requirement for our algorithm.
The goal in DP is to compute the initial value function $J_0(x)$ defined in~\eqref{eq_recursion} and the corresponding optimal policy $\pi^\star_0(x)$ given in~\eqref{eq_recursion_policy} for all $x \in \X^d$. Depending on the application, it often suffices to calculate $J_0(x_0)$ and $\pi^\star_0(x_0)$ for a specific initial state $x_0 \in \X^d$. 

\paragraph{Stochastic setting.}
A natural generalization of the DP setting is to allow non-deterministic elements, which greatly expands its applicability. Let $\xi$ be a discrete random variable with a known probability mass function $p_{\xi}$ that has a support of size $r$. We consider stochastic dynamics of the form $f(x,u)=A'x+B'u+\xi$, generalizing Assumption~\ref{ass:modelling}~\eqref{it_dynamics}. We choose this specific form of stochastic dynamics because it preserves convexity of the DP operator, and it is sufficient for many applications, where it allows modeling a random exogenous shock affecting the state transition; for example, in inventory management applications it allows modeling an uncertain product demand.
The DP operator is defined analogously to~\eqref{eq_DP} by adding an expectation, i.e.,
\begin{align} \label{eq_DP_stoch}
J_t(x)=\mathrm{DP_{\!stoch}}[J_{t+1}](x) :=g_{\mathrm{x}}(x)+ \min_{u \in \U^c} \Big \{g_{\mathrm{u}}(u) + \E\big[ \,J_{t+1}\big(A'x+B'u+\xi\big)\big] \Big \} \quad \forall x \in \X^d \, .
\end{align}
As in the deterministic setting~\eqref{eq_recursion}, the value function is expressed as
\begin{align*}
    J_0(x)= \mathrm{DP}^T[J_T](x)  \qquad \textnormal{where} \quad J_T(x)=g_T(x) \quad \forall x \in \X^d\, .
\end{align*}
The corresponding optimal policy is given by
\begin{align*}
    \pi^\star_0(x)=\arg\min_{u \in \U^d}\Big\{ g_{\mathrm{u}}(u) + \E\big[J_0(A'x+B'u+\xi)\big] \Big\} \, .
\end{align*}
In the stochastic setting we may assume a cost function $g_{\mathrm{u}}$ that depends on the random variable $\xi$. This can be incorporated into the setting above by a slight abuse of notation by letting $g_{\mathrm{u}}(u):= \E[g_{\mathrm{u}}(u,\xi)]$, hence we always write simply $g_{\mathrm{u}}$ in the following.  
Furthermore, we suppose that a feasibility assumption similar to Assumption~\ref{ass:feasibility} holds.
\begin{assumption}[Feasibility]\label{ass:feasibility_stoch}\
\begin{enumerate}[(i)]
\item\label{it_ass1_stoch} Original problem: For every value function $J: \X^d \to \R$ and for every $x \in \X^d$ there exists $u_x^\star \in \arg\min_{u\in\U^{c}} \{g_{\mathrm{u}}(u) + \E[J(A'x + B'u+\xi)]\} $ such that $A'x+B'u_x^\star+\xi \in \X^d$ for every $\xi$.
\item\label{it_ass2_stoch} Conjugate problem: For every value function $J': \cX_N^d \to \R$ and for every $x \in \cX_N^d$ there exists $u_x^\star \in \arg\min_{u\in\U^{c}} \{g_{\mathrm{u}}(u) + \E[\bar {J'}(A'x + B'u+\xi)]\} $ such that $A'x+B'u_x^\star+\xi \in \X^d$ for every $\xi$.
\end{enumerate}
\end{assumption}

Given the structure of the DP problem studied in this paper, we can exploit the concept of a \emph{post-decision state}. More precisely, in our problem the transition probabilities can be defined in terms of an auxiliary state variable, called the post-decision state, defined as $m_t := A'x_t + B' u_t$~\cite[Chapter~6]{ref:Bertsekas-12}. The corresponding post-decision state space is defined as $\M^d:=\{A'x + B'u : x \in \X^d, u \in \U^c\}=:\Q^{d_{\dr}}\times \cR_{P_{\ci}}^{d_{\ci}}$.\footnote{By definition the cardinality of the discrete part of the post-decision space can be bounded by $P_{\ci}\leq M_{\ci} N_{\ci}$.}
The transition probabilities are denoted by
\begin{align*}
 \PP[x_{t+1}=j | x_t = i, u_t=u]
    = \PP[\xi = j - A'i-B'u]
    = \PP[\xi=j-m]
    = p_{\xi}(j-m) \, ,
\end{align*}
where $m=A'i+B'u$.
We define the optimal cost-to-go at the post-decision state $m$ as
\begin{align*} 
    V_t(m)
    := \int_{\X^d} p_{\xi}(x-m) J_t(x) \di x
    = \sum_{k=0}^{r-1} p_{\xi}(\xi_k) J_t(m+\xi_k) \, ,
\end{align*}
where $\xi_0,\ldots,\xi_{r-1}$ denote the $r$ possible realizations of the random variable $\xi$.
Altogether, the concept of post-decision states allows us to consider an equivalent model for the DP problem, with a different state space, which can be advantageous for computation.
Specifically, for all $m \in \M^d$ the cost-to-go function $V_t(m)$ satisfies the recursion
\begin{align}
    V_t(m) 
    &= \sum_{k=0}^{r-1} p_{\xi}(\xi_k) \mathrm{DP}[V_{t+1}](m+\xi_k) \nonumber \\
    &=\sum_{k=0}^{r-1} p_{\xi}(\xi_k) \left( g_{\mathrm{x}}(m+\xi_k) + \min_{u \in \U^c}\{  g_{\mathrm{u}}(u) + V_{t+1}\big(A'(m+\xi_k )+B'u\big) \} \right) \nonumber \\
    &=:\widehat{\mathrm{DP}}[V_{t+1}](m) \, , \label{eq_DPhat}
\end{align}
which is a DP recursion over post-decision states, rather than over the original state space.
The concept of the post-decision state is powerful since the optimal policy to the DP is given by the simple formula
\begin{align*}
    \pi_0^\star(x)=\arg \min_{u \in \U^c}\big \{ g_{\mathrm{u}}(u) + V_0(A'x+B'u)\big\} \, .
\end{align*}

When comparing this DP equation over post-decision states~\eqref{eq_DPhat} to the DP equation over the standard states~\eqref{eq_DP_stoch}, we observe that the expectation and minimization operators are swapped. This structure will be exploited by the quantum algorithm in Section~\ref{sec_quantum_algo}.
As in the deterministic setting, we define a shifted stochastic DP operator as
\begin{align}\label{eq_DP_shifted_stoch}
  \mathrm{\widehat{DP}_{shift}}[V](m)  :=  \widehat{\mathrm{DP}}[V](m) - \E[ g_{\mathrm{x}}(m+\xi)] \quad \forall m \in \M^d \, .
\end{align}

\subsection{Previous work \& hardness of dynamic programming} \label{sec_previous_work}
In the deterministic discrete setting ($d_\dr = c_\dr = 0$), the standard DP algorithm~\cite{bertsekas1995dynamic} computes $J_0(x)$, as well as $\pi^\star_0(x)$, for all $x\in \cX_N^d$ in time $O(T N_\ci M_\ci)$. This is not a polynomial-time algorithm, as even in simple cases (e.g., the state and action spaces are a subset of $\mathbb{Z}$) the running time is pseudopolynomial. In the deterministic continuous setting ($d_\ci = c_\ci = 0$),
the DP problem can be solved in time $O(\poly(d))$, because under Assumption \ref{ass:modelling}, we can rewrite~\eqref{eq_recursion} as a convex optimization problem, that can be solved efficiently with standard methods~\cite{boyd_book}.

The complexity of the quantum algorithm developed in this paper is polylogarithmic in the size of the state space, and exponential in $d$ and $T$. In special cases, thoroughly discussed in Section~\ref{sec_quantum_algo}, the base of the exponential can be one, yielding a polynomial-time algorithm to construct a quantum-mechanical representation of the value function; the running time increases if we want to extract a classical description, and hence our results do not contradict any widely believed complexity-theoretical assumptions. Classically, the convex DP problem is generally easy when both $d$ and $T$ are fixed, because it can be formulated as a mathematical optimization problem with a fixed number of variables and constraints. Its solution in polynomial time is then possible with standard convex optimization techniques in the purely continuous case, or with Lenstra-type algorithms in the integer and mixed-integer case, provided the cost functions are polynomial~\cite{koppe2012complexity}. In the stochastic case, however, the formulation as a mathematical optimization problem has size that depends on the support of the random variables, and to the best of our knowledge the computational complexity of this case has not been determined. When $d$ is not fixed, even the purely discrete convex DP is NP-hard, as the closest vector problem can be trivially reduced to it. In the stochastic case, this paper shows that when $d=1$, the convex purely continuous DP problem is already \#P-hard if $T$ is not fixed, and~\cite{halman2014fully} shows a similar result for the purely discrete case.

Given the wide applicability of the DP framework, finding quantum algorithms for DP with a provable speedup has been a major open question in theoretical computer science. The fact that standard DP algorithms work sequentially, first solving small subproblems and then iteratively increasing the problem size, is generally viewed as an obstacle in deriving good quantum algorithms for DP. We also remark that under the widely believed conjecture $\text{NP} \nsubseteq \text{BQP}$~\cite{BV97}, we can expect at most a polynomial quantum speedup for DP. The paper~\cite{ambainis19_2} presents an algorithm based on Grover search that provides a speedup for some NP-hard problems whose best classical algorithm is an exponential time application of DP. The algorithm decreases the running time from $\tilde O(2^n)$ to $\tilde O(1.728^n)$;\footnote{The $\tilde O(\cdot)$ notation ignores logarithmic factors. The same convention holds for $\tilde \Omega(\cdot)$ that is used later in the manuscript.} this represents an important advance in a field where progress is rare, although it is not a fully-quantum algorithm, but rather a hybrid scheme that applies classical DP with Grover acceleration in several parts of the algorithm. The paper~\cite[Theorem~IV.1]{pooya19} shows that for general DP problems (i.e., not necessarily satisfying Assumption~\ref{ass:modelling}), the speedup that quantum algorithms can achieve is at most quadratic in the number of states and in the size of the decision space, but an algorithm that matches this bound is not known.\footnote{The quantum algorithm presented in \cite{pooya19} is known to contain an error that invalidates the main proof.} 

\subsection{Results}
We study a class of DP problems that we call \emph{convex DP problems}. A convex DP problem satisfies Assumptions~\ref{ass:modelling} and~\ref{ass:feasibility}, and in addition, the value function is convex or convex extensible at every stage; in other words,  the shifted DP operator defined in~\eqref{eq_DP_tilde} (and in~\eqref{eq_DP_shifted_stoch} for the stochastic setting) preserves convexity. Note that Assumptions~\ref{ass:modelling} and~\ref{ass:feasibility} by themselves imply convexity of the value function at every stage only for purely continuous problems (see Lemma~\ref{lem_cont_convexity_preserving}). For discrete DP problems, a few sufficient conditions for the convexity of the value function are known, most notably when the cost functions are $L^\natural$-convex, or the state space is one-dimensional. Problems satisfying these assumptions have many applications in operations management, see~\cite{halman2014fully,chen2015natural}. We discuss these sufficient conditions in more detail in Section~\ref{sec_classical}. Let $\cX^d_N=\{x_0,\ldots,x_{N-1}\}$ denote a discretization of the state space $\X^d$ with $N=N_{\dr}N_{\ci}$ and $N_{\dr} \sim (T/\eps)^{d_{\dr}}$. The main contributions of this paper are summarized as follows:
\begin{itemize}
\item For the deterministic case, we present an algorithm that computes a quantum-mechanical approximation of the value function, i.e., a state $\ket{\psi}=\frac{1}{\sqrt{N}} \sum_{i=0}^{N-1}\ket{i} \ket{\hat J_0(x_i)}$ such that for any $\eps>0$
\begin{align*} 
  |\hat J_0(x_i) -  J_0(x_i)| \leq \eps \quad \forall i\in \{0,\ldots, N-1\} \, ,
\end{align*}
in time $O(T \gamma^{dT}\polylog(N_{\ci},(T/\eps)^d))$, where the parameter $\gamma$ is a problem specific parameter depending on the condition numbers of the functions $\bar g_{\mathrm{x}},\bar g_{\mathrm{u}}$, and $\bar J_T$. For some DP problems, e.g., separable quadratic cost functions with condition number $1$, we have $\gamma=1$ (see Remark~\ref{rmk_gamma}). 
We refer to Theorem~\ref{thm_QDP} for a more precise statement.
\item The presented algorithm (combined with amplitude amplification) can output an $\eps$-approximation of $J_0(x_i)$, for any $i \in \{0,\ldots,N-1\}$, in time $O(T \gamma^{dT}(T/\eps)^{d_{\dr}/2} \sqrt{N_\ci}\, \polylog(N_\ci,(T/\eps)^d))$. The corresponding optimal policy $\pi_0^\star(x_i)$ can be approximately computed by additionally solving a convex optimization problem over the decision space.
We refer to Corollary~\ref{cor_eval_J0} for the details.
\item The quantum algorithm can be extended to the stochastic setting with discrete random variables, and the running time increases by a factor $O(r \gamma^r)$, where $r$ is the maximum size of the support of the random variables (see Corollary~\ref{cor_QDP_stoch}).
\item We show that the class of problems to which our quantum algorithm applies contains $\#$P-hard problems with $\gamma=1$ (see Proposition~\ref{prop_hardDP}). Furthermore, we show that for quadratic continuous stochastic DPs with $\gamma=1$ the quantum algorithm achieves a quadratic speedup compared to the best-known classical algorithm.
\item We show that in an oracle setting the quantum algorithm is optimal (up to polylogarithmic factors) for problems with $\gamma=1$ (see Section~\ref{sec_oracle_optimal}). This is a consequence of lower bounds for the computation of the QLFT proven in \cite[Section 6]{QLFT20}. 
\end{itemize}
This work presents the first quantum algorithm for solving the class of convex DP problems that achieves a quadratic speedup in certain cases; moreover, it does so while relying on the Bellman backward recursion. To the best of our knowledge, classical lower bounds for convex DP problems of the type studied in this paper are not known: this leaves open the possibility that there exist faster classical algorithms than the Bellman recursion. However, we believe that the quadratic speedup shown in this paper may be of interest to the research community even if faster classical algorithms are eventually discovered. To obtain a quantum speedup in the context of the Bellman backward recursion, we leverage the quantum speedup of the QLFT. More specifically, the QLFT exploits quantum superposition to obtain a provable quantum speedup over any classical algorithm to compute the LFT in some scenarios; this is shown in \cite{QLFT20}, and summarized in Section \ref{sec_qlft}. Our DP algorithm is essentially the quantization of a classical LFT-based algorithm for DP, but the running time analysis and proof of correctness require many intermediate technical results. 

We leave it as an open question to determine if the convexity assumptions can be relaxed while still keeping the quadratic speedup. This would be important as many combinatorial optimization problems formulated in the framework of DP (e.g., knapsack, traveling salesman) do \emph{not} have a convex value function, and it therefore remains an open question to find a quantum DP approach that is faster than classical algorithms for those problems. However, we note that this may require different techniques compared to the ones introduced in this work.

\section{Preliminaries}
\subsection{Notation}
For $N \in \N$ we denote $[N]:=\{0,1,\ldots,N-1\}$. 
For a vector $x=(x_0,\ldots,x_{N-1})$ its Euclidean and maximum norm are denoted by by $\norm{x}$ and $\norm{x}_{\infty}$, respectively. For $1\leq k \leq \ell \leq n$ we denote $x_k^\ell=(x_{k},x_{k+1},\ldots,x_{\ell})$. Analogously, we write $\ket{x_k^\ell}$ for the computational basis state $\ket{x_{k},x_{k+1},\ldots,x_{\ell}}$, and $\ket{f(x_k^\ell)}$  as a shorthand for $\ket{f(x_{k}),f(x_{k+1}),\ldots,f(x_{\ell})}$.
A function $f:\R^d \to \R$ is called Lipschitz continuous with constant $L_f \geq 0$ if $|f(x)-f(y)|\leq L_f \norm{x-y}$ for all $x,y \in \R^d$. 
The function is said to be $\mu_f$-strongly convex if for all $x,y \in \R^d$ and $t\in [0,1]$ we have $f(tx +(1-t)y)\leq t f(x) +(1-t)f(y) +\frac{1}{2}\mu_f t(1-t)\norm{x-y}^2$. 
The condition number of a strongly convex function $f$ is defined by $\kappa_f:=L'/\mu$~\cite[Section~2.1.3]{ref:nesterov-book-04}, where $L'$ denotes the Lipschitz constant of $\nabla f$. We note that by definition $\kappa_f \geq 1$. 
To simplify notation, we call a function $f$ that is convex extensible Lipschitz continuous with constant $L_f$ and/or $\mu_f$-strongly convex if their convex extension $\bar f$ satisfies these properties. Analogously we say that $f$ has condition number $\kappa_f$ meaning that $\bar f$ has condition number $\kappa_{\bar f}$.
The indicator function is defined by $\mathds{1}\{X\}:=1$ if $X=\mathrm{true}$ and $0$ otherwise.
The logical 'and' and 'or' operations are denoted by $\wedge$ and $\vee$, respectively. 
For a set $\cS$ let $\Delta_{\mathcal{S}}:= \max\{ \sup_{s,s'\in\mathcal{S}}\|s-s'\|, \sup_{s\in \cS}\norm{s}\}$. The one-sided Hausdorff distance between two sets $\cX$ and $\cY$ is denoted by $\distH(\cX,\cY) := \sup_{x\in \cX}\inf_{y\in \cY} \|x-y\|$. 
\subsection{Discrete convex functions}
We now define convex extensible functions, a concept that is used in Assumption~\ref{ass:modelling}. A function defined over the integers $J':\cX^d_N\to\R$, where $\cX^d_N=\{x_0,\ldots,x_{N-1}\}$, is called \textit{convex extensible} if there exists a convex function $\bar J':\X^d\to\R$ such that $J'(x) = \bar J'(x)$ for all $x\in\cX^d_N$. Following \cite[Section~3.4]{ref:Murota-03}, the convex extension $\bar J'$ can be defined as 
\begin{equation} \label{eq:convex:extension}
\bar J'(x) = \sup_{p\in\R^d, \alpha \in \R} \left\{ \inprod{p}{x}+\alpha \ : \ \inprod{p}{y}+\alpha \leq J'(y) \ \forall y\in \cX^d_N \right\}.
\end{equation}
Several classes of discrete functions exist that are known to be convex extensible: convex seprable functions; $L^\natural$-convex functions (``L'' stands for ``Lattice''); and $M^\natural$-convex functions (``M'' stands for ``Matroid''), see \cite{ref:Murota-03} for definitions and details. These concepts are relevant for this paper only insofar as they are helpful to give sufficient conditions for the value function to be convex extensible over the integers.

\subsection{Quantum Legendre-Fenchel transform}
\label{sec_qlft}
For a function $f: \X^d \to \R$, its \emph{Legendre-Fenchel transform} (LFT), also known as \emph{convex conjugate} or simply \emph{Legendre transform}, is the function $f^{\ast}: \mathbb{S}^d \to \R$ defined as
\begin{align}  \label{eq_LFT}
f^{\ast}(s) := \sup_{x \in \X^d}\{\langle s, x \rangle - f(x) \} \, ,
\end{align}
where $\mathbb{S}^d$ denotes the dual space that is chosen to be a compact and convex subset of $\R^d$.
We use the shorthand notation $f^{\ast}=\cL_{x\to s}(f)$. It is important to remark that the LFT is usually defined over a topological vector space, but in our case, $\X^d$ may not even be a compact set when $d_i > 0$. We take \eqref{eq_LFT} to be the definition of the LFT for any set $\X^d$.  In Section~\ref{sec_classical} we discuss a classical algorithm for DP that is based on the LFT, and that is the foundation for our quantum algorithm. We prove all the standard properties of the LFT that are necessary for our algorithm to work, ensuring that they hold for the type of sets $\X^d$ (i.e., mixed-integer sets) and functions $f$ (i.e., convex extensible) considered in this paper; intuitively, this is the case because we are working with convex extensible functions, ensuring that we can equivalently work with a continuous convex function that matches the original function at the discrete points.
If the primal and dual spaces $\X^d$ and $\mathbb{S}^d$ are replaced by discrete sets $\cX_N^d$ and $\cS_K^d$ of size $N$ and $K$, respectively, the transform analogous to~\eqref{eq_LFT} is called \emph{discrete LFT} and defined as
\begin{align} \label{eq_dLFT}
f^{*}(s) := \max_{x \in \cX_N^d}\{\langle s, x \rangle - f(x) \} \, ,
\end{align}
for $s \in \cS_K^d$.
We use the two different asterisk symbols \raisebox{-0.425mm}{\large{$\ast$}} and $*$ to distinguish between the continuous and the discrete LFT and refer the interested reader to~\cite{QLFT20} for more details about the discrete LFT.
The following results summarize some properties of the LFT and in particular the difference between the discrete and the continuous LFT.
\begin{lemma}\label{lem_contLFT}
Let $f,g: \X^d \to \R$ be such that $|f(x)-g(x)|\leq \eps$ for all $x\in \X^d$. Then $|f\!\protoast\!(s)-g\!\protoast\!(s)| \leq \eps$ for all $s \in \mathbb{S}^d$ and $|f^*(s)-g^*(s)| \leq \eps$ for all $s \in \cS_K^d$.
\end{lemma}
\begin{proof}
Let $(x_n)_{n\in \N}$ be a maximizing sequence of $f\!\protoast\!(s)$, i.e., $\lim_{n\to \infty}( \inprod{x_n}{s}-f(x_n))=f\!\protoast\!(s)$.
By definition we have 
\begin{align*}
f\!\protoast\!(s)-g\!\protoast\!(s)
\leq \lim_{n \to \infty} \big( g(x_n)-f(x_n) \big)
\leq \eps \, .
\end{align*}
The same argument shows $g\!\protoast\!(s)-f\!\protoast\!(s)\leq \eps$ which proves the assertion. The statement for the discrete LFT follows analogously.
\end{proof}
\begin{lemma} \label{lem_discreteLFT_approx}
Let $f:\X^d \to \R$ be Lipschitz continuous with constant $L_f$ and let $\cX_N^d$ and $\cS_K^d$ denote the primal and dual spaces of the discrete LFT. Then,
\begin{enumerate}[(i)]
\item\label{it_LFT1} $|f\!\protoast\!(s) - f^*(s)| \leq (1+\sqrt{d})L_f \distH(\X^d,\cX_N^d)$ for all $s \in \cS_K^d$;
\item\label{it_LFT2} $f\protoast\!$ and $f^*$ are Lipschitz continuous with constants $\Delta_{\X^d}$ and $\Delta_{\cX_N^d}$, respectively;
\item\label{it_LFT3} $|f\!\protoast\!\protoast(x) - f^{**}(x)|\leq  (1+\sqrt{d})\Delta_{\X^d} \distH(\mathbb{S}^d,\cS_K^d) + (1+\sqrt{d})L_{f} \distH(\mathbb{X}^d,\cX_N^d)$ for all $x \in \cX_N^d$;
\item\label{it_LFT4} $f\!\protoast\!\protoast(x)=f(x)$ for all $x \in \X^d$ if $f$ is proper convex or convex extensible.
\end{enumerate}
\end{lemma}
The proof is given in Appendix~\ref{app_pfLemmadiscLFT}. We end the discussion on properties of the LFT by recalling the following statements that are straightforward to verify.
\begin{fact} \label{fact_LFT_basic}
Let $f,g : \X^d \to \R$. Then,  
\begin{enumerate}[(i)]
    \item $f \leq g$ implies $f\protoast \geq g\protoast$; \label{it_reverse}
    \item $g(x):=f(x)+\alpha$ for $\alpha \in \R$ implies $g\!\protoast\!(s)=f\!\protoast\!(s)-\alpha$;\label{it_shift}
\end{enumerate}
\end{fact}

In a recent paper~\cite{QLFT20}, we introduced a quantum algorithm called \emph{quantum Legendre-Fenchel transform} (QLFT) to efficiently compute the discrete LFT in superposition. We give an intuitive description of the algorithm and refer to~\cite{QLFT20} for more details. Given a discrete dual space $\cS^d_K$, to compute the discrete LFT at every dual point we need to determine the maximizer of~\eqref{eq_dLFT}. If the function $f$ is sufficiently well-behaved, we show in \cite{QLFT20} that every point in primal space $\cX^d_N$ is a maximizer for exactly the same number of dual points. In this case, starting from a superposition of $x \in \cX^d_N$, the computation of $f^*(s)$ for all $s \in \cS^d_K$ amounts to simple relabeling of the indices followed by trivial computations. In reality, the calculations are more involved because computing the relabeling requires some work, and we must be able to relax the assumption that each point in primal space is a maximizer of \eqref{eq_dLFT} for the same number of $s \in \cS^d_K$. It is shown in~\cite[Proposition~6.1]{QLFT20} and \cite[Proposition~6.3]{QLFT20} that any classical algorithm to output the value of the discrete LFT of a $d$-dimensional function at an arbitrary dual point $s$ has time complexity $\tilde{\Omega}(2^d)$, and furthermore, that even sampling $f^*(s)$ for $s$ chosen uniformly at random in $\cS^d_K$ has the same complexity. On the other hand, for a function $f$ with condition number $\kappa$, the running time of the QLFT scales as $\tilde O(\kappa^{d/2})$ --- an exponential advantage when $\kappa = 1$, but this only accounts for the time to prepare the answer in superposition. If we want to output $f^*(s)$ at a given $s \in \cS^d_K$, the quantum algorithm is only quadratically faster than classical.

The quantum algorithm for DP introduced in this paper exploits the quantum speedup of the QLFT. Here, we need a variant of the QLFT algorithm presented in~\cite{QLFT20}, which we describe next. We state the result for univariate functions; using the factorization property of the LFT we can reduce the $d$-dimensional case to $d$ times the one-dimensional scenario. We remark that the quantum speedup of the QLFT comes precisely from exploiting superposition in the $d$-dimensional case: for univariate functions, there also exists an efficient classical algorithm. For a set $\cX_N=\{x_0,\ldots,x_{N-1}\}$, define the discrete gradients as
\begin{align*}
    c_{i} := \frac{f(x_{i+1})-f(x_i)}{x_{i+1}-x_i}  \quad \textnormal{for } i \in [N-1] \, .
\end{align*}
\begin{proposition}[Variant of the QLFT~\cite{QLFT20} for one-dimensional case] \label{prop_QLFT}
Let $f:\X\to \R$ be a function that is convex and has a condition number $\kappa_f<\infty$.
Let $k,n,\ell \in \N$, $N=2^{n}$, $K=2^k$, $\cX_N=\{x_0,\ldots,x_{N-1}\}$ such that $x_{i+1}=\delta_{\mathrm{x}}+x_i$ for some $\delta_{\mathrm{x}}\geq 0$, and $\cS_K=\{s_0,\ldots,s_{K-1}\}$ such that $s_0=c_0$, $s_{K-1}=c_{N-2}$ and $s_{j+1}=\delta_{\mathrm{s}}+s_j$ for some $\delta_{\mathrm{s}}\geq 0$.
There is an algorithm that performs the transformation\footnote{Elements of the vector $x$ that are outside of $[N]$ are irrelevant and thus can be set to an arbitrary value.}
\begin{align} \label{eq_iterative_regular_QLFT}
        \frac{1}{\sqrt{N}} \sum_{i=0}^{N-1} \ket{i}  \ket{f(x_{i-\ell -1}^{i+\ell+1})} \ket{\textnormal{Garbage}(i)} \quad \to \quad \frac{1}{\sqrt{K}} \sum_{j=0}^{K-1} \ket{j} \ket{f^*(s_{j-\ell}^{j+\ell})}\ket{\textnormal{Garbage}(j)} \, .
\end{align}
Combined with amplitude amplification, the expected running time for constant probability of success is $O(\sqrt{\kappa_f} \, \polylog(N,K))$.
\end{proposition}
The above result follows from a small modification of~\cite[Proof of Theorem~4.7]{QLFT20}. The main difference compared to \cite{QLFT20} is that we simultaneously compute the discrete LFT for $f$ evaluated at multiple consecutive points, while keeping the same probability of success. The proof is given in Appendix~\ref{app_proof_QLFT}. 
The algorithm can be extended to the $d$-dimensional setting, as discussed in~\cite[Section~5]{QLFT20}, with a running time \smash{$O(\kappa_f^{d/2} \, \polylog(N,K))$}.
If the function $f$ is not differentiable but it is convex extensible, we can use $\kappa$ of the convex extension to determine the running time; if it is not convex extensible, then the running time of the QLFT algorithm from Proposition~\ref{prop_QLFT} for dimension $d$ is $O((NW/K)^{d/2} \polylog(N,K))$, where the parameter $W$ is defined in~\eqref{eq_W}.

\section{Dynamic programming via the Legendre-Fenchel transform} \label{sec_classical}
In this section we present an approach to approximately compute the DP operator using the discrete LFT. This approach, inspired by~\cite{peyman}, serves as the starting point to derive the quantum algorithm in Section~\ref{sec_quantum_algo}. 

\subsection{Deterministic setting}
 Let \smash{$J:\Y^{d_{\dr}} \times \cZ^{d_{\ci}}_{N_{\ci}} \to\R$} be convex in $y$ and convex extensible in $z$, let \smash{$J':\cY_{N_{\dr}}^{d_{\dr}} \times \cZ^{d_{\ci}}_{N_{\ci}} \to\R$} be convex extensible such that $J(y,z)=J'(y,z)$ $\forall y\in \cY_{N_{\dr}}^{d_{\dr}}, z\in \cZ^{d_{\ci}}_{N_{\ci}}$.
Consider a discrete dual space $\cS_{K_{\dr}}^{d_{\dr}}\times {\cS'}_{K_{\ci}}^{d_{\ci}}$. Then, we define the conjugate DP operator as
\begin{align} \label{eq_conjDPoperator}
   \mathrm{DP_{\!conj}}[J'](y,z) := g_{\mathrm{x}}(y,z) +  h^*(Ay,Dz) \quad \forall y\in \mathcal{Y}_{N_{\dr}}^{d_{\dr}},z\in \cZ^{d_{\ci}}_{N_{\ci}} \, , 
\end{align}
where $h(s,s'):= g_{\mathrm{u}}\hspace{-1.5mm}\protoast \!(-B^\top s,-C^\top s - E^\top s') + {J'}^*(s,s')$ $\forall s\in\mathcal{S}_{K_{\dr}}^{d_{\dr}},s'\in\mathcal{S'}_{K_{\ci}}^{d_{\ci}}$ as visualized in Figure~\ref{fig_classical_idea} and formally stated in Algorithm~\ref{algo_CDP}.
\begin{figure}[!htb]
    \begin{center}
    \begin{tikzpicture}
    \def \x{6}
    \def \y{2}
    \def \la{0.3}      
    \node at (0,0) {$J'(y, z)$};
    \node at (0,-\y) {${J'}^{*}(s, s')$};
    \node at (\x,0) {$h^{*}(Ay, Dz) $};
    \node at (\x,-\y) {$h(s, s')$};
    \node at (2*\x-3*\la,0) {$\mathrm{DP_{\!conj}}[J'](y, z)$};  
    \draw[thick,->,gray] (0,-\la) -- (0,-\y+\la);  
    \draw[thick,->,gray] (3*\la,-\y) -- (\x-2*\la,-\y);  
    \draw[thick,<-,gray] (\x,-1.5*\la) -- (\x,-\y+\la);  
    \draw[thick,->,gray] (3.5*\la+\x,0) -- (2*\x-7.4*\la,0);   
    \node[gray] at (-2.5*\la,-0.5*\y)  {$\cL_{y,z\to s,s'}$};   
    \node[gray] at (\x/2,-\y+\la)  {$+g_{\mathrm{u}}\hspace{-2.5mm}\protoast \!(-B^\top s, -C^\top s -E^\top s')$};       
    \node[gray] at (\x-3.5*\la,-0.5*\y)  {$\cL_{s,s' \to Ay, Dz}$};     
    \node[gray] at (3*\x/2-2*\la,\la)  {$+g_{\mathrm{x}}(y, z)$};        
    \end{tikzpicture}
    \end{center}
\caption{Graphical visualization of Algorithm~\ref{algo_CDP}. It can be seen that the only operations to compute the conjugate DP operator are the discrete LFT and simple additions.}
\label{fig_classical_idea}
\end{figure}

The conjugate DP operator enjoys many useful properties: (a) it can be easily computed via the discrete LFT, under the assumption that $g_{\mathrm{u}}\hspace{-1.5mm}\protoast \,$ is available; (b) it preserves convex extensibility, i.e., it maps convex extensible functions to convex extensible functions; (c) it approximates the DP operator defined in~\eqref{eq_DP}, where the approximation error can be controlled via the size of the primal and dual discrete spaces. The first property is obvious; the second and third property will be shown in the following.
Indeed, we next show that $\mathrm{DP_{\!conj}}[J'] - g_{\mathrm{x}}$ approximates the biconjugate $\mathrm{DP_{shift}}[J]\!\protoast\!\protoast$, which is the convexification of \smash{$\mathrm{DP_{shift}}[J]$}. Via~\eqref{eq_DP_tilde}, the operator $\mathrm{DP_{\!conj}}[J']$ can then be interpreted as a convex approximation to $\mathrm{DP}[J]$. The definition of $J$ and $J'$ in the statement below stems from our goal to accommodate the case in which we are given a function $J$ defined on a continuous or mixed-integer domain, and we discretize the domain, thereby having access to a function $J'$ defined only on the discretization.

\begin{algorithm}[!htb]
\caption{Calculation of the conjugate DP operator~\eqref{eq_conjDPoperator}}
\label{algo_CDP}
\begin{algorithmic}
\STATE \textbf{Input:} $\cY_{N_{\dr}}^{d_{\dr}}=\{y_0,\ldots,y_{N_{\dr}-1}\}$, $\cZ^{d_{\ci}}_{N_{\ci}}=\{z_0,\ldots,z_{N_{\ci}-1}\}$, $\cS_{K_{\dr}}^{d_{\dr}}=\{s_0,\ldots,s_{K_{\dr}-1}\}$, ${\cS'}_{K_{\ci}}^{d_{\ci}}=\{s'_0,\ldots,s'_{K_{\ci}-1}\}$, $J'(\cdot)$, $g_{\mathrm{x}}(\cdot)$, $g_{\mathrm{u}}\hspace{-2.5mm}\protoast \!(\cdot)$ \vspace{2mm} \\
\item \textbf{Output:} $\mathrm{DP_{\!conj}}[J'](y,z)$ $\forall y\in \mathcal{Y}_{N_{\dr}}^{d_{\dr}},~z\in \cZ^{d_{\ci}}_{N_{\ci}}$
\begin{enumerate}
\item Compute ${J'}^*(s,s') \quad \forall s\in\mathcal{S}_{K_{\dr}}^{d_{\dr}},~s'\in\mathcal{S'}_{K_{\ci}}^{d_{\ci}}$;
\item Let $h(s,s')= g_{\mathrm{u}}\hspace{-2.5mm}\protoast \!(-B^\top s,-C^\top s - E^\top s') + {J'}^*(s,s') \quad \forall s\in\mathcal{S}_{K_{\dr}}^{d_{\dr}},~s'\in\mathcal{S'}_{K_{\ci}}^{d_{\ci}}$;
\item Compute $h^*(Ay,Dz) \quad \forall y\in \mathcal{Y}_{N_{\dr}}^{d_{\dr}},~z\in \cZ^{d_{\ci}}_{N_{\ci}}$;
\item Let $\mathrm{DP_{\!conj}}[J'](y,z) := g_{\mathrm{x}}(y,z) +  h^*(Ay,Dz) \quad \forall y\in \mathcal{Y}_{N_{\dr}}^{d_{\dr}},~z\in \cZ^{d_{\ci}}_{N_{\ci}}$; 
\end{enumerate}
Output $\mathrm{DP_{\!conj}}[J'](y,z)$
\end{algorithmic}
\end{algorithm}

\begin{theorem}[Properties of $\mathrm{DP_{\!conj}}$] \label{thm_classical_error}
Consider a DP problem satisfying Assumptions~\ref{ass:modelling} and~\ref{ass:feasibility}.
Let \smash{$J:\Y^{d_{\dr}} \times \cZ^{d_{\ci}}_{N_{\ci}} \to\R$} be convex in $y$ and convex extensible in $z$, let \smash{$J':\cY_{N_{\dr}}^{d_{\dr}} \times \cZ^{d_{\ci}}_{N_{\ci}} \to\R$} be convex extensible such that $J(y,z)=J'(y,z)$ $\forall y\in \cY_{N_{\dr}}^{d_{\dr}}, z\in \cZ^{d_{\ci}}_{N_{\ci}}$, and let $\cS_{K_{\dr}}^{d_{\dr}}\times {\cS'}_{K_{\ci}}^{d_{\ci}}=:\bar\cS^d_K$ be discrete dual spaces. Then,
\begin{enumerate}[(i)]
\item\label{it_convexExtensible} $\mathrm{DP_{\!conj}}[J']$ is convex extensible, i.e., the conjugate DP operator preserves convex extensibility;
\item\label{it_Lipschitz} $\mathrm{DP_{\!conj}}[J']$ is Lipschitz continuous with constant $\sqrt{d}L_{J'}\max\{\norm{A}_{\infty},\norm{D}_{\infty}\} + L_{g_{\mathrm{x}}}$; 
\item\label{it_errorBound} for all $y\in \cY_{N_{\dr}}^{d_{\dr}}, z\in \cZ^{d_{\ci}}_{N_{\ci}}$ we have    
    \begin{align} \label{eq_errorBound}
	 |\mathrm{{DP}_{\!conj}}[J'](y,z) - \mathrm{DP_{shift}}[J]\!\protoast\!\protoast(y,z) - g_{\mathrm{x}}(y,z) | \leq E_1 + E_2\, ,
    \end{align}
with error terms $E_1  := (1+\sqrt{d}) L_J\, \distH(\Y^{d_{\dr}},\cY_{N_{\dr}}^{d_\dr})$ and $E_2 := (1+\sqrt{d})(\tau +\eta)\distH(\mathbb{S}^d,\bar\cS_K^d)$, where the constants $\tau$ and $\eta$ are defined in Assumption~\ref{ass:modelling}, and $\mathbb{S}^d \subseteq \R^d$ is a compact convex space.
\end{enumerate}
\end{theorem}
The proof of Theorem~\ref{thm_classical_error} is given in Appendix~\ref{app_pf_thm_error_bound}.
The error bounds of Theorem~\ref{thm_classical_error} require the Lipschitz continuity of the function $J$. 
In the following remark we argue that this naturally holds.
\begin{remark}[Lipschitz continuity] \label{rmk_Lipschitz}
The $\mathrm{DP}$ operator~\eqref{eq_DP} preserves Lipschitz continuity. One can easily show that
\begin{align*}
\mathrm{DP}[J](x) - \mathrm{DP}[J](x') &\leq g_{\mathrm{x}}(x) -  g_{\mathrm{x}}(x') + J(A'x + B'u^\star) - J(A'x'+B'u^\star) \\
&\leq L_{g_{\mathrm{x}}} \|x-x'\| + L_J \|A'x - A'x'\| \\
&\leq (L_{g_{\mathrm{x}}} + L_J \norm{A'}_{\infty})\|x-x'\| \, .
\end{align*}
Hence, if $J$ is $L_J$-Lipschitz continuous, then $\mathrm{DP}[J]$ is $(L_{g_{\mathrm{x}}} + L_J \norm{A'}_{\infty})$-Lipschitz continuous. Theorem~\ref{thm_classical_error} shows that the conjugate DP operator also preserves Lipschitz continuity. 
\end{remark}
\begin{remark}[Scaling of discretization error] \label{rmk_scaling}
When considering the approximation error of Theorem~\ref{thm_classical_error}, it is natural to ask how the two error terms $E_1$ and $E_2$ scale in terms of the discretization granularity, i.e, in $N_{\dr}$ and $K$. Note that in case of a regular discretization scheme, it can be seen that $\distH(\Y^{d_{\dr}},\cY_{N_{\dr}}^{d_{\dr}})$ scales proportionally to $(1/N_{\dr})^{1/d_{\dr}}$, and analogously $\distH(\mathbb{S}^{d},\cS_{K}^{d})$ scales proportionally to $(1/K)^{1/d}$. In other words, provided that $J$ is Lipschitz continuous (Remark~\ref{rmk_Lipschitz}), we see that 
\begin{align*}
E_1 \sim (1/N_{\dr})^{1/d_{\dr}} \qquad \textnormal{and} \qquad E_2 \sim (1/K)^{1/d} \, ,
\end{align*}
which implies $\eps_{\mathrm{disc}}:= E_1 + E_2 \leq \eps$ when choosing
\begin{align} \label{eq_scaling_param}
N_{\dr} \sim (1/\eps)^{d_{\dr}} \qquad \textnormal{and} \qquad K \sim (1/\eps)^d \, .
\end{align}
\end{remark}

Theorem~\ref{thm_classical_error} proves that the scheme from Figure~\ref{fig_classical_idea} is able to approximate a single step of the backward recursion for DP. The scheme also allows us to approximately compute the optimal policy, as shown in the next result. Recall that in the scheme depicted in Figure~\ref{fig_classical_idea} we have
\begin{align} \label{eq_def_sStar}
   h^*(A'x) = \max_{s \in \bar\cS^d_K} \{ \inprod{A'x}{s} - J^*(s) - g_{\mathrm{u}}\hspace{-2.5mm}\protoast \!(-{B'}^\top s)  \} \, .
\end{align}
\begin{lemma}[Optimal policy]\label{lem_optimal_policy}
Let $J$ and $J'$ be as in Theorem~\ref{thm_classical_error} and consider a DP satisfying Assumptions~\ref{ass:modelling} and~\ref{ass:feasibility}, for which the shifted DP operator \smash{$\mathrm{\widetilde{DP}}$} preserves convexity.
Let $\cX_N^d, \cS_K^d$ be primal and dual spaces such that, using Theorem~\ref{thm_classical_error}, we have $|\mathrm{DP}[J](x)-\mathrm{DP}_{\mathrm{conj}}[J'](x)|\leq \eps$ for all $x\in\cX_N^d$ and $\eps>0$. For any $x\in\cX_N^d$, let
\begin{align}
 \hat \pi(x) = \arg\min_{u \in \U^c} \Big\{g_{\mathrm{u}}(u) + \inprod{u}{{B'}^\top s_x^{\star}}  \Big \} \, ,    \label{eq_optPolicy_via_LFT_robust}
\end{align}
where $s_x^{\star}$ denotes the optimizer in~\eqref{eq_def_sStar}. Then, $\norm{\pi^{\star}(x) - \hat \pi(x)} \leq \sqrt{4\eps/\mu_{ g_{\mathrm{u}}}}$, where $\pi^{\star}(x)$ is the optimal action.
\end{lemma}
The proof is given in Appendix~\ref{app_pf_optimal_policy}. 
Using the Lipschitz continuity of $J$ and $g_{\mathrm{u}}$ we can quantify the quality of the value obtained by taking the approximately optimal decision $\hat \pi(x)$ as
\begin{align*}
|J\big(A'x + B' \hat \pi(x)\big) + g_{\mathrm{u}}\big(\hat \pi(x)\big) + g_{\mathrm{x}}(x) - \mathrm{DP}[J](x)| 
    &\leq L_{J} \norm{B'\big(\hat \pi(x)-\pi^\star(x)\big)} + L_{g_{\mathrm{u}}} \norm{\hat \pi(x)-\pi^\star(x)}\\
    &\leq (L_{J} \norm{B'}_{\infty} + L_{g_{\mathrm{u}}})\sqrt{4\eps/\mu_{g_{\mathrm{u}}}} \, ,
\end{align*}
where the final step uses Lemma~\ref{lem_optimal_policy}.

At this point we have all the necessary components to describe the LFT-based classical algorithm to solve convex DP problems.
\begin{corollary} \label{cor_cDP_optimal}
Given the setting as in Lemma~\ref{lem_optimal_policy}. Let ${\mathrm{DP}^T_{\!\mathrm{conj}}}[J']$ be the function obtained after applying $T$ times Algorithm~\ref{algo_CDP}. Then, for every \smash{$y\in \cY_{N_{\dr}}^{d_{\dr}}, z\in \cZ^{d_{\ci}}_{N_{\ci}}$}
  \begin{align} 
	 |{\mathrm{DP}^T_{\!\mathrm{conj}}}[J'](y,z) - \mathrm{DP}^T[J](y,z)| \leq T(E_1 + E_2)\, ,
  \end{align}
where $E_1$ and $E_2$ are defined as in Theorem~\ref{thm_classical_error}.
\end{corollary}
\begin{proof}
  The proof is by induction on the backward recursion. The first step of the induction is given by Theorem~\ref{thm_classical_error} plus~\eqref{eq_DP_tilde}. For the induction step, after each application of $\mathrm{DP_{\!conj}}$ we obtain a convex extensible approximation of the shifted DP operator, and accumulate an error upper bounded by $E_1 + E_2$. Since the shifted DP operator preserves convexity by assumption, $\mathrm{DP_{shift}}[J] = \mathrm{DP_{shift}}[J]\!\protoast\!\protoast$. The induction step is thus proven.
\end{proof}

Algorithm~\ref{algo_CDP} and Figure~\ref{fig_classical_idea} show that the basic operation to compute the conjugate DP operator is the discrete LFT. For the generalization to the quantum case in the next section we need to understand how the conjugate DP operator changes the condition number. This is done with the following lemma, which is proven in Appendix~\ref{app_proof_lemma_condNumber}.
\begin{lemma}[Condition number of value function] \label{lem_condNumber_valueFunction}
Consider a DP problem satisfying Assumptions~\ref{ass:modelling} and~\ref{ass:feasibility} and the iteration $J'_t(x):=\mathrm{DP_{\!conj}}[J'_{t+1}](x)$ for $t=T-1,\ldots,0$.
Then,
\begin{align} \label{eq_condNumber_ineq}
    \kappa_{J'_t} 
    \leq \phi(t,T,L'_{g_{\mathrm{x}}},\mu_{g_{\mathrm{x}}},L'_{g_{\mathrm{u}}},\mu_{g_{\mathrm{u}}},L'_{J_T},\mu_{ J_T}) \, ,
\end{align}
where the function $\phi$ is defined in~\eqref{eq_phi_def} in the proof.
\end{lemma}
The bound from Lemma~\ref{lem_condNumber_valueFunction} can be tight, i.e., there exist problems where~\eqref{eq_condNumber_ineq} holds with equality. For example in scenarios where $\kappa_{ g_{\mathrm{x}}}=\kappa_{g_{\mathrm{u}}}=\kappa_{J_{T}}=1$ we have \smash{$\phi(t,T,L'_{g_{\mathrm{x}}},\mu_{ g_{\mathrm{x}}},L'_{g_{\mathrm{u}}},\mu_{g_{\mathrm{u}}},L'_{J_T},\mu_{J_T})=1$} for all $t\in [T]$ and hence~\eqref{eq_condNumber_ineq} is an equality. We do not give a precise definition of $\phi$ in the main text because it is cumbersome, but all details can be found in Appendix~\ref{app_proof_lemma_condNumber}.

In Theorem~\ref{thm_classical_error} we have seen that the conjugate DP operator preserves convex extensiblity of a function.
Similar to the standard DP operator~\eqref{eq_DP}, we can show that the conjugate DP operator~\eqref{eq_conjDPoperator} is also closed under certain subclasses of convex extensible functions; in particular, we show this for $L^\natural$ convex functions under some conditions on the system dynamics.

\begin{remark}[Conjugate DP operator preserves $L^\natural$-convexity] \label{rmk_Lnatural}
Suppose Assumptions~\ref{ass:modelling} and~\ref{ass:feasibility} hold, the function $\varphi(s)=g_{\mathrm{u}}\hspace{-1.5mm}\protoast (-B^\top s,-C^\top s-E^\top s')$ is separable convex,\footnote{A function $f:\cS_K^d \to \R \cup \{\infty\}$ is called \emph{separable convex} if it can be represented as $f(s)=\sum_{i=0}^{d-1} f_i(s_i)$, where $s=(s_i)_{i=0}^{d-1} \in \cS_K^d$ and $f_i:\cS_K \to \R \cup \{\infty\}$ is a univariate discrete convex function, see~\cite[p.~95]{ref:Murota-03}.} and $A,D \propto \mathds{1}$. If $J$ is $L^\natural$-convex, then $\mathrm{DP_{\!conj}}[J]$ is $L^\natural$-convex.
To see this, we recall the definition of $\mathrm{DP_{\!conj}}$ given in~\eqref{eq_conjDPoperator} and visualized by Algorithm~\ref{algo_CDP} and Figure~\ref{fig_classical_idea}. Given that $J:\cX_N^d\to\R$ is $L^\natural$-convex, its discrete Legendre-Fenchel dual $J^*$ is known to be $M^\natural$-convex \cite[Theorem~8.12]{ref:Murota-03}; and the sum of an $M^\natural$-convex function and a separable convex function is again $M^\natural$-convex \cite{ref:Murota-99}. Hence, the function $h$, defined in Algorithm~\ref{algo_CDP} as $h(s,s')= g_{\mathrm{u}}\hspace{-1.5mm}\protoast \!(-B^\top s,-C^\top s-E^\top s') + J^*(s,s')$ is $M^\natural$-convex.
The discrete Legendre-Fenchel transform $h^*$ of an $M^\natural$-convex function is an $L^\natural$-convex function~\cite[Theorem~8.12]{ref:Murota-03}. This implies that also the function $(y,z)\mapsto h^*(Ay,Dz)$ is $L^\natural$-convex, because $A,D \propto \mathds{1}$ \cite{ref:Murota-19}. Finally, as the sum to two $L^\natural$-convex functions is $L^\natural$-convex, we obtain that $\mathrm{DP_{\!conj}}[J](y,z) := g_{\mathrm{x}}(y,z) +  h^*(Ay,Dz)$ is $L^\natural$-convex.
\end{remark}
Note that a similar statement holds for the standard DP operator, as we discuss next; this property has been exploited in some approximate schemes for DP, e.g., in~\cite{chen2014fixed}.

\begin{remark}[DP operator preserves $L^\natural$-convexity] \label{rmk_Lnatural_DP}
Assume that the matrices $A',B' \propto \mathds{1}$ and that the functions $g_{\mathrm{x}}:\cX^d_N\to\mathbb{R}$, $g_{\mathrm{u}}:\cU^c_M\to\mathbb{R}$, $J:\cX^d_N\to\mathbb{R}$ are $L^\natural$-convex. Then, $\mathrm{DP}[J]$ is $L^\natural$-convex. To see this, recall that the sum of two $L^\natural$-convex functions is again $L^\natural$-convex \cite{ref:Murota-19} and that since $A',B' \propto \mathds{1}$ we have that the mapping $(x,u)\mapsto g_{\mathrm{u}}(u) + J(A'x + B'u)$ is $L^\natural$-convex. Since partial minimization preserves $L^\natural$-convexity \cite[Proposition~4.10]{ref:Murota-19} and since $g_{\mathrm{x}}$ is $L^\natural$-convex, the function $x\mapsto \mathrm{DP}[J](x) = g_{\mathrm{x}}(x) + \min_{u\in\cU^c_M}\{g_{\mathrm{u}}(u) + J(A'x + B'u)\}$ is $L^\natural$-convex.
\end{remark}

Since $L^\natural$-convex functions are convex extensible, the above discussion shows that LFT-based algorithm can be applied on purely discrete problems with $L^\natural$-convex terminal cost function $g_T$. DP with $L^\natural$-convex functions find applications in operations management, see \cite{chen2015natural} and the references therein. Another class of problems for which the DP operator is known to be convex at every stage is the one described in \cite{halman2014fully}, which also lists some applications; problems in this class have $d = 1, c = 1$, convex cost functions (for univariate functions, convexity for discrete functions coincides with convex extensibility), and transition function of the form $f(x, u) = \alpha_x x + \alpha_u u + \alpha_c$, where $\alpha_x \in \mathbb{Z}$, $\alpha_u \in \{-1,0,1\}$ and $\alpha_c \in \mathbb{Z}$.

We end this section with a technical lemma ensuring that for the purely continuous case, convexity of the cost function suffices to guarantee that the shifted DP operator is convexity presererving.
\begin{lemma} \label{lem_cont_convexity_preserving}
In an purely continuous setting, the shifted DP operator defined in~\eqref{eq_DP_tilde} is convexity preserving. 
\end{lemma}
\begin{proof}
Note that under Assumption~\ref{ass:modelling} the shifted DP operator~\eqref{eq_DP_tilde} in the purely continuous setting can be expressed as
\begin{align} \label{eq_DP:ass}
\mathrm{DP_{shift}}[J] = \min_{v \in \V^{c_{\dr}}} \{g_{\mathrm{u}}(v) + J \big( Ay+Bv \big) \} \quad \forall y \in \Y^{d_{\dr}} \, .
\end{align}
In a first step, we claim that the function $(y,v)\mapsto g_{\mathrm{u}}(v) + J \big( Ay+Bv \big)$ is jointly convex. Since the function $g_{\mathrm{u}}$ is convex by Assumption~\ref{ass:modelling}, it remains to show that $(y,v)\mapsto J(Ay + Bv)$ is jointly convex, which follows from the convexity of $J$. Indeed for any $y_1,y_2\in\Y^{d_{\dr}}$, $v_1,v_2\in\V^{c_{\dr}}$ and $\lambda\in[0,1]$, we have
\begin{align*}
J(A(\lambda y_1+(1-\lambda)y_2) + B (\lambda v_1+(1-\lambda)v_2)) 
&=J(\lambda(Ay_1 + Bv_1) + (1-\lambda)(Ay_2 + Bv_2))\\
&\leq \lambda J(Ay_1 + Bv_1) + (1-\lambda) J(Ay_2 + Bv_2), 
\end{align*}
where the inequality follows from the convexity of $J$. Therefore,~\eqref{eq_DP:ass} is the partial minimum of a jointly convex function, which is known to be convex~\cite[Section~3.2.5]{boyd_book}.
\end{proof}

\subsection{Stochastic setting} \label{sec_classical_stochastic}
We next show how to modify the scheme in the previous section to approximate the stochastic DP operator for the post-decision state, defined in~\eqref{eq_DPhat}. Let $V': \cM^d_{P} \to \R$ be a convex extensible function and let $V:\M^d\to \R$ be its convex extension.
The conjugate stochastic DP operator is defined as
\begin{align*}
 \mathrm{\widehat{DP}_{conj}}[V](m):= \sum \limits_{k=0}^{r-1} p_{\xi}(\xi_k)\Big(h^{*}\big(A'(m+\xi_k)\big)+g_{\mathrm{x}}(m+\xi_k) \Big)\, ,
\end{align*}
for $m=(q,r)$ and $h(s,s'):= g_{\mathrm{u}}\hspace{-2.5mm}\protoast \!(-B^\top s,-C^\top s - E^\top s') + {V'}^*(s,s')$. Figure~\ref{fig_classical_idea_stochastic} visualizes the definition of the operator \smash{$\mathrm{\widehat{DP}_{conj}}[V]$}.

\begin{figure}[!htb]
    \begin{center}
    \begin{tikzpicture}
    \def \x{6}
    \def \y{2}
    \def \la{0.3}      
    \node at (-1,0) {$V(q,r)$};
    \node at (-1,-\y) {${V}^{*}(s,s')$};
    \node at (3*\x/4,0) {$h^{*}(A'm) $};
    \node at (3*\x/4,-\y) {$h(s,s')$};
    \node at (2*\x-6.5*\la,0) {$\underbrace{\sum \limits_{k=0}^{r-1} p_{\xi}(\xi_k)\Big(h^{*}\big(A'(m+\xi_k)\big)+g_{\mathrm{x}}(m+\xi_k) \Big)}_{\mathrm{\widehat{DP}_{conj}}[V](q,r)} $};  
    \draw[thick,->,gray] (-1,-\la) -- (-1,-\y+\la);  
    \draw[thick,->,gray] (-1+2.5*\la,-\y) -- (-1+3*\x/4+1*\la,-\y);  
    \draw[thick,<-,gray] (3*\x/4,-1.5*\la) -- (3*\x/4,-\y+\la);  
    \draw[thick,->,gray] (3*\la+3*\x/4,0) -- (\x+2.5*\la,0);   
    \node[gray] at (-0.5*\la,-0.5*\y)  {$\cL_{q,r\to s,s'}$};   
    \node[gray] at (3*\x/8-1.5*\la,-\y+\la)  {$+g_{\mathrm{u}}\hspace{-2.5mm}\protoast \!(-B^\top s,-C^\top s-E^\top s')$};       
    \node[gray] at (\x-1.5*\la,-0.5*\y)  {$\cL_{s,s' \to A'm}$};     
    \end{tikzpicture}
    \end{center}
\caption{Modification of the scheme in Figure~\ref{fig_classical_idea} for the stochastic setting where $m=(q,r)$ and $x=(y,z)$.}
\label{fig_classical_idea_stochastic}
\end{figure}

The following corollary is the adaptation of Theorem~\ref{thm_classical_error} to the stochastic setting, showing that $\mathrm{\widehat{DP}_{conj}}[V]$ is a good approximation to $\mathrm{\widehat{DP}}[V]$ as long as the shifted DP operator defined in~\eqref{eq_DP_shifted_stoch} preserves convexity.
\begin{corollary}[Properties of $\mathrm{\widehat{DP}_{conj}}$] \label{cor_stoch_error}
Consider a stochastic DP problem satisfying Assumptions~\ref{ass:modelling} and~\ref{ass:feasibility_stoch}.
Let $V':\cQ^{d_{\dr}}_{P_{\dr}} \times \cR^{d_{\ci}}_{P_{\ci}} \to\R$ be a convex extensible function, let $V:\Q^{d_{\dr}} \times \cR^{d_{\ci}}_{P_{\ci}} \to \R$ be 
convex in $q$ and convex extensible in $r$ such that \smash{$V'(q,r)=V(q,r)$ $\forall q \in \cQ^{d_{\dr}}_{P_{\dr}}, r \in \cR^{d_{\ci}}_{P_{\ci}}$},
and let \smash{$\cS_{K_{\dr}}^{d_{\dr}}\times {\cS'}_{K_{\ci}}^{d_{\ci}}=:\bar \cS^d_K$} be a discrete dual space. Then,
\begin{enumerate}[(i)]
    \item\label{it_convexExtensible_stoch} $\mathrm{\widehat{DP}_{\!conj}}[V']$ is convex extensible, i.e., the conjugate DP operator preserves convex extensibility;
 \item\label{it_Lipschitz_stoch} $\mathrm{\widehat{DP}_{\!conj}}[V']$ is Lipschitz continuous with constant $\sqrt{d}L_{V}\max\{\norm{A}_{\infty},\norm{D}_{\infty}\} + L_{g_{\mathrm{x}}}$; 
    \item\label{it_errorBound_stoch} for all $m \in \cM_{P}^{d}$ we have
    \begin{align*}
	\big|\mathrm{\widehat{DP}_{conj}}[V'](m) - \mathrm{\widehat{DP}_{shift}}[V]\!\protoast\!\protoast(m) -\E[ g_{\mathrm{x}}(m+\xi)] \big| \leq E_1 + E_2\, , 
\end{align*}
with error terms $E_1  := (1+\sqrt{d}) L_J\, \distH(\Q^{d_{\dr}},\cQ_{N_{\dr}}^{d_\dr})$ and $E_2 := (1+\sqrt{d})(\tau +\eta)\distH(\mathbb{S}^d,\bar\cS_K^d)$, where the constants $\tau$ and $\eta$ are defined in Assumption~\ref{ass:modelling} and $\mathbb{S}^d\subseteq \R^d$ is a compact convex space.
\end{enumerate}
\end{corollary}
\begin{proof}
The three statements follow from the proof of Theorem~\ref{thm_classical_error} by recalling that $\widehat{\mathrm{DP}}[V]$ can be viewed as a convex combination of $\mathrm{DP}[V]$ evaluated at different points, as shown in~\eqref{eq_DPhat}.
\end{proof}
We can apply Corollary~\ref{cor_stoch_error} recursively, in the same fashion as in Corollary~\ref{cor_cDP_optimal}, to approximate the initial value function $V_0(m)$ for any $m \in \M^d$. Similarly to the deterministic case, we can also approximately compute the corresponding optimal policy $\pi_0^\star(m)$. 
To see this, recall that in the scheme depicted in Figure~\ref{fig_classical_idea_stochastic} we have
\begin{align} \label{eq_def_sStar_stoch}
   h^*(A'm) = \max_{s \in \bar\cS^d_K} \{ \inprod{A'm}{s} - V^*(s) - g_{\mathrm{u}}\hspace{-2.5mm}\protoast \!(-{B'}^\top s)  \} \, .
\end{align}

\begin{corollary}[Optimal policy]\label{cor_optimal_policy_stoch}
Let $V$ and $V'$ be as in Corollary~\ref{cor_stoch_error} and consider a DP satisfying Assumptions~\ref{ass:modelling} and~\ref{ass:feasibility_stoch}, for which the shifted DP operator $\mathrm{\widehat{DP}_{shift}}$ preserves convexity. 
Let $\cM_P^d,\cS_K^d$ be primal and dual spaces such that, using Corollary~\ref{cor_stoch_error}, we have $|\mathrm{\widehat{DP}}[V](m)-\mathrm{\widehat{DP}_{\!conj}}[V'](m)|\leq \eps$ for all $m \in \cM^d_P$ and $\eps>0$. 
For any $m\in\cM_P^d$, let
\begin{align}
    \hat \pi(m) = \arg\min_{u \in \U^c} \Big\{g_{\mathrm{u}}(u) + \inprod{u}{{B'}^\top s_m^{\star}}  \Big \} \, ,  \label{eq_optPolicy_via_LFT_robust_stoch}
\end{align}
where $s_m^{\star}$ denotes the optimizer in~\eqref{eq_def_sStar_stoch}.
Then, $\norm{\pi^{\star}(m) - \hat \pi(m)} \leq \sqrt{4\eps/\mu_{g_{\mathrm{u}}}}$, where $\pi^{\star}(m)$ is the optimal action.
\end{corollary}
\begin{proof}
Follows the same steps as in the proof of Lemma~\ref{lem_optimal_policy}.
\end{proof}


\section{Quantum algorithms for dynamic programming} \label{sec_quantum_algo}
In this section we present a quantum algorithm for DP. 
The classical DP algorithm based on the Bellman equations outputs a vector $(J_0(x_0),\ldots,J_0(x_{N-1}))$, containing the initial value function, and the corresponding optimal policies $(\pi_0^\star(x_0),\ldots,\pi_0^\star(x_{N-1}))$, or a way to compute them (see Lemma \ref{lem_optimal_policy}). By contrast, the quantum algorithm outputs a quantum-mechanical representation of the value function. More precisely, we aim to construct the quantum state
\begin{align} \label{eq_superpos_valueFunction}
\frac{1}{\sqrt{N}} \sum_{i=0}^{N-1}\ket{i} \ket{J_0(x_i)} \, .
\end{align}
We are also interested in constructing the state \smash{$\frac{1}{\sqrt{N}} \sum_{i=0}^{N-1}\ket{i} \ket{J_0(x_i)}\ket{\pi_0^\star(x_i)}$}, or alternatively, some quantum state containing information that allows us to recover the optimal policy, similarly to Lemma \ref{lem_optimal_policy}. After creating the state~\eqref{eq_superpos_valueFunction} with a unitary operation, we can apply Grover search~\cite{grover96} to evaluate the value function at any fixed point, i.e., output $J_0(x_i)$ or the pair $(J_0(x_i),\pi^\star_0(x_i))$ at any fixed $x_i \in \cX_N^d$ for $\cX_N^d=\{x_0,\ldots,x_{N-1}\}$, with $O(\sqrt{N})$ applications of the unitary to create the state. This is discussed in Corollary~\ref{cor_eval_J0}.

In certain scenarios it may be helpful to consider a different representation of the value function. We can apply a quantum digital-analog conversion~\cite[Theorem~1]{kosuke19} to the state~\eqref{eq_superpos_valueFunction} to obtain
\begin{align} \label{eq_superpos_amplitudes}
  \ket{J_0}:= \frac{1}{\alpha} \sum_{i=0}^{N-1} J_0(x_i) \ket{i} \, .
\end{align}
This can be done with a quantum algorithm that has an expected running time  $O(\sqrt{\omega}\,\polylog(N))$, where \smash{$\omega:=\alpha/(N \max_{i\in[N]} J_0(x_i)^2)$ for $\alpha:= \sum_{i=0}^{N-1}J_0(x_i)^2$}. If the numbers $J_0(x_0),\ldots,J_0(x_{N-1})$ are sufficiently uniformly distributed, the parameter $\omega$ does not scale with $N$~\cite[Remark~4.4]{QLFT20}.
The analog representation~\eqref{eq_superpos_amplitudes} can be used to efficiently evaluate certain functions of the complete value function. For example, if $H \in \C^{N\times N}$ is an observable that can be implemented with complexity $O(\polylog(N))$, we can approximately compute the expectation value $\bra{J_0} H \ket{J_0}$ in polylogarithmic time.

We next discuss two assumptions that we require to make the quantum algorithm efficient. Both of them are standard in the literature and not too restrictive in practice.
\begin{assumption}[Access to functions $g_{\mathrm{x}}$, $g_{\mathrm{u}}\hspace{-1.5mm}\protoast\,$, and $J_T$] \label{ass_Uf}
We assume that we have access to unitaries $U_{g_{\mathrm{x}}}$, $U_{g_{\mathrm{u}}\hspace{-1.5mm}\protoast\,}$, and $U_{J_T}$ such that
\begin{align*}
U_{g_{\mathrm{x}}}(\ket{x_i}\ket{0})= \ket{x_i}\ket{g_{\mathrm{x}}(x_i)}, \quad 
U_{g_{\mathrm{u}}\hspace{-1.5mm}\protoast\,}(\ket{s_j}\ket{0})= \ket{s_j}\ket{g_{\mathrm{u}}\hspace{-1.9mm}\protoast\!(s_j)}, \quad \textnormal{and} \quad
U_{J_T}(\ket{x_i}\ket{0})= \ket{x_i}\ket{J_T(x_i)} \, ,
\end{align*}
for every given $(x_0,\ldots,x_{N-1})$ and $(s_0,\ldots,s_{K-1})$. Furthermore, the cost of running $U_{g_{\mathrm{x}}}$, $U_{J_T}$, and $U_{g_{\mathrm{u}}\hspace{-1.5mm}\protoast\,}$ is  $O(\polylog(N))$ and $O(\polylog(K))$, respectively.\footnote{In case of a stochastic DP setting the terminal value function $J_T$ is replaced with $V_T$.}
\end{assumption}
The assumption is justified because for every function that can be efficiently computed with a classical algorithm, i.e., computable in time $O(\polylog(N))$, we can use quantum arithmetic to load the data efficiently~\cite{nielsenChuang_book}. For the second mapping above, we require that $g_{\mathrm{u}}$ is sufficiently well-behaved that its continuous LFT $g_{\mathrm{u}}\hspace{-1.5mm}\protoast\,$ features a closed form expression.
\begin{assumption}[Sufficient precision] \label{ass_precision}
We assume to have sufficient precision such that all basic quantum arithmetic operations can be executed without any errors.
\end{assumption}
The assumption of sufficient precision is necessary because otherwise, the discrete LFT could be affected by errors that are hard to quantify. We remark that the classical discrete LFT faces the same difficulty, see \cite{QLFT20} for a discussion. In other words, this assumption is necessary to facilitate the algorithm analysis in both the classical and quantum setting. It is not a restrictive assumption because the running time of the algorithm is polynomial in the number of (qu)bits, hence we can increase precision at a small cost.
\subsection{Deterministic setting} \label{sec_det_quantum}
We now present a quantum algorithm, Algorithm~\ref{algo_QDP}, that computes a superposition of the initial value function, in the form of~\eqref{eq_superpos_valueFunction}, for a deterministic convex DP problem. We consider fixed primal and dual spaces $\cX^d_N=\{x_0,\ldots,x_{N-1}\}$ and $\cS^d_K=\{s_0,\ldots,s_{K-1}\}$ that are discretized with regular steps.\footnote{In Algorithm~\ref{algo_QDP} we consider vectors with $4T$ elements because each of the $T$ iteration steps consists of two QLFT steps and in each QLFT we loose two elements.} Since $\cX^d_N$ and $\cS^d_K$ are multidimensional regular grids, we assume that a superposition of the corresponding points can be constructed in polylogarithmic time: this is natural, as the coordinates of each grid point can be computed with simple quantum arithmetics.

\begin{algorithm}[!htb]
\caption{Deterministic convex QDP}
\label{algo_QDP}
\begin{algorithmic}
\STATE \textbf{Input:} $T \in \N$, $(y_0,\ldots,y_{N_{\dr}-1})$, $(z_0,\ldots,z_{N_{\ci}-1})$, $(s_0,\ldots,s_{K_{\dr}-1})$, $(s'_0,\ldots,s'_{K_{\ci}-1})$, $o_i=-B^\top s_i$, $p_i=-C^\top s_i$ $p'_i=-E^\top s'_i$, $q_i=Ay_i$, $r_i=Dz_i$, and terminal value function $J'_T=:\hat J_T$; \vspace{2mm} \\
\item \textbf{Output:} Approximation to $\frac{1}{\sqrt{N}} \sum_{i=0}^{N_{\dr}-1}\sum_{i'=0}^{N_{\ci}-1}\ket{i,i'} \ket{J_0(y_i,z_{i'})}\ket{\textnormal{Garbage}(i,i')}$; \vspace{2mm} \\
Prepare $\frac{1}{\sqrt{N}} \sum_{i=0}^{N_{\dr}-1}\sum_{i'=0}^{N_{\ci}-1} \ket{i,i'} \ket{\hat J_T(y_{i-2T}^{i+2T},z_{i'-2T}^{i'+2T})}$; \\
\item \textbf{For } $\ell=T,\ldots,1$ \textbf{ do} \\
\begin{enumerate}
\item $\frac{1}{\sqrt{K}} \sum\limits_{j=0}^{K_{\dr}-1}\sum\limits_{j'=0}^{K_{\ci}-1} \ket{j,j'} \ket{\hat{J}^*_{\ell}(s_{j-2\ell+1}^{j+2\ell-1},{s'}_{j'-2\ell+1}^{j'+2\ell-1})} \ket{\textnormal{Garbage}(j,j')}$;
\item $\frac{1}{\sqrt{K}}\! \sum\limits_{j=0}^{K_{\dr}-1}\sum\limits_{j'=0}^{K_{\ci}-1}\! \ket{j,j'}  \ket{\underbrace{\hat{J}^*_{\ell}(s_{j-2\ell+1}^{j+2\ell-1},{s'}_{j'-2\ell+1}^{j'+2\ell-1})\!+\!g_{\mathrm{u}}\hspace{-2.5mm}\protoast \!(o_{j-2\ell+1}^{j+2\ell-1},p_{j-2\ell+1}^{j+2\ell-1}+{p'}_{j'-2\ell+1}^{j'+2\ell-1})}_{=:h_{\ell}(s_{j-2\ell+1}^{j+2\ell-1},{s'}_{j'-2\ell+1}^{j'+2\ell-1})} } \ket{\textnormal{Garbage}(j,j')}$;
\item $\frac{1}{\sqrt{N}} \sum\limits_{i=0}^{N_{\dr}-1}\sum\limits_{i'=0}^{N_{\ci}-1} \ket{i,i'} \ket{h^*_{\ell}(q_{i-2(\ell-1)}^{i+2(\ell-1)},r_{i'-2(\ell-1)}^{i'+2(\ell-1)})} \ket{\textnormal{Garbage}(i,i')}$;
\item $\ket{\psi_{\ell}}= \frac{1}{\sqrt{N}} \sum\limits_{i=0}^{N_{\dr}-1}\sum\limits_{i'=0}^{N_{\ci}-1} \ket{i,i'}  \ket{\underbrace{h^*_{\ell}(q_{i-2(\ell-1)}^{i+2(\ell-1)},r_{i'-2(\ell-1)}^{i'+2(\ell-1)}) + g_{\mathrm{x}}(y_{i-2(\ell-1)}^{i+2(\ell-1)},z_{i'-2(\ell-1)}^{i'+2(\ell-1)})}_{=:\hat{J}_{\ell-1}(y_{i-2(\ell-1)}^{i+2(\ell-1)},z_{i'-2(\ell-1)}^{i'+2(\ell-1)})} } \ket{\textnormal{Garbage}(i,i')}$;
\end{enumerate}
\textbf{end}\\
Output $\ket{\psi_1}= \frac{1}{\sqrt{N}} \sum_{i=0}^{N_{\dr}-1}\sum_{i'=0}^{N_{\ci}-1} \ket{i,i'}\ket{\hat{J}_0(y_i,z_{i'})} \ket{\textnormal{Garbage}(i,i')}$;
\end{algorithmic}
\end{algorithm}

\begin{theorem}[Deterministic convex QDP] \label{thm_QDP}
Let $\eps>0, T \in \N$, $\cY_{N_{\dr}}^{d_{\dr}}=\{y_0,\ldots,y_{N_{\dr}-1}\}$ be a regular discretization with $N_{\dr}\sim (T/\eps)^{d_{\dr}}$, and consider a DP problem satisfying Assumptions~\ref{ass:modelling} and~\ref{ass:feasibility} where the shifted DP operator defined in~\eqref{eq_DP_tilde} preserves convexity.
Given Assumption~\ref{ass_precision}, the output $\ket{\psi_1}=\frac{1}{\sqrt{N}} \sum_{i=0}^{N-1}\ket{i} \ket{\hat{J}_0(x_i)}\ket{\textnormal{Garbage}(i)}$ for a successful run of Algorithm~\ref{algo_QDP} satisfies
\begin{align*}
|\hat{J}_0(x_i) - J_0(x_i)| \leq  \eps \quad \forall i \in [N] \, ,
\end{align*}
where $\cY^{d_{\dr}}_{N_{\dr}} \times \cZ^{d_{\ci}}_{N_{\ci}} = \cX^d_N=\{x_0,\ldots,x_{N-1}\}$.
Combined with amplitude amplification and given Assumption~\ref{ass_Uf}, the expected running time for constant probability of success is 
\begin{align} \label{eq_runtimeQDP}
    O\left(T \gamma^{dT}\polylog\big(N_{\ci},(T/\eps)^{d}\big)\right) \, ,
\end{align}
with $\gamma:=\phi(0,T,L'_{g_{\mathrm{x}}},\mu_{g_{\mathrm{x}}},L'_{g_{\mathrm{u}}},\mu_{g_{\mathrm{u}}},L'_{J_T},\mu_{J_T})$, where the function $\phi$ is defined in Lemma~\ref{lem_condNumber_valueFunction}.
\end{theorem}
The proof is given in Appendix~\ref{app_pf_QDP}.
With the help of Lemma~\ref{lem_optimal_policy}, Algorithm~\ref{algo_QDP} can be modified to output a state $\ket{\psi'}=\frac{1}{\sqrt{N}} \sum_{i=0}^{N-1}\ket{i} \ket{\hat{J}_0(x_i)} \ket{\hat\pi_0(x_i)}\ket{\textnormal{Garbage}(i)}$ such that
\begin{align} \label{eq_variant_QDP}
|\hat{J}_0(x_i) - J_0(x_i)| \leq  \eps \qquad \textnormal{and} \qquad \norm{\hat\pi_0(x_i)-\pi_0^\star(x_i)} \leq  \sqrt{\frac{4\eps}{\mu_{g_{\mathrm{u}}}}} \qquad \forall i \in [N]\,,
\end{align}
at the additional cost of solving the convex optimization problem~\eqref{eq_optPolicy_via_LFT_robust} on top of the running time stated in~\eqref{eq_runtimeQDP}. More precisely, during the final QLFT step we can  keep track of the optimizers $\{s_i^\star\}_{i=1}^N$. We can then classically solve  problem~\eqref{eq_optPolicy_via_LFT_robust}, or  use a quantum algorithm for the same task, to obtain the approximation $\hat \pi_0(x_i)$ to $\pi^\star_0(x_i)$. Note that all the data required to solve \eqref{eq_optPolicy_via_LFT_robust} are readily available.

Algorithm~\ref{algo_QDP} computes an approximation of the state~\eqref{eq_superpos_valueFunction}, i.e., a superposition of the value function (and the optimal policy) at $T=0$. The next corollary discusses how to output the value function (and the optimal policy) at any specific point, using amplitude amplification.
\begin{corollary}[Evaluating value function and optimal policy at specific points] \label{cor_eval_J0}
Let $\eps>0, T \in \N$, $\cY_{N_{\dr}}^{d_{\dr}}=\{y_0,\ldots,y_{N_{\dr}-1}\}$ be a regular discretization with $N_{\dr}\sim (T/\eps)^{d_{\dr}}$, and consider a DP problem satisfying Assumptions~\ref{ass:modelling} and~\ref{ass:feasibility} where the shifted DP operator defined in~\eqref{eq_DP_tilde} preserves convexity.  
Given Assumptions~\ref{ass_Uf} and~\ref{ass_precision}, for any $i \in [N]$ Algorithm~\ref{algo_QDP} combined with amplitude amplification outputs $\hat{J}_0(x_i) \in \R$ such that
\begin{align*}
|\hat{J}_0(x_i)  - J_0(x_i)| \leq \eps\, ,
\end{align*}
with an expected running time
\begin{align} \label{eq_runtime_J0x}
    O\left(T \gamma^{d T}(T/\eps)^{d_{\dr}/2} \sqrt{N_{\ci}}\, \polylog\big(N_{\ci},(T/\eps)^d\big) \right) \, ,
\end{align}
for $\gamma:=\phi(0,T,L'_{g_{\mathrm{x}}},\mu_{g_{\mathrm{x}}},L'_{g_{\mathrm{u}}},\mu_{g_{\mathrm{u}}},L'_{J_T},\mu_{J_T})$, where the function $\phi$ is defined in Lemma~\ref{lem_condNumber_valueFunction}.
Furthermore, Algorithm~\ref{algo_QDP} can compute $\hat \pi_0(x_i)$ such that
\begin{align*}
    \norm{\hat \pi_0(x_i)-\pi_0^{\star}(x_i)} \leq \sqrt{\frac{4\eps}{\mu_{g_{\mathrm{u}}}}} \, ,
\end{align*}
in a running time that additionally to~\eqref{eq_runtime_J0x} requires the time to solve~\eqref{eq_optPolicy_via_LFT_robust}, where $\pi_0^{\star}(x_i)$ denotes the optimal initial policy for state $x_i$.
\end{corollary}
\begin{proof}
We simply postpone the postselection steps in the QLFT (see~\cite[Step~3 in Algorithm~3]{QLFT20}) to the very end, and amplify the projection of the quantum algorithm onto the ``good subspace'', i.e., the space that we want to postselect on (given by the indicator function in~\cite[Step~3 in Algorithm~3]{QLFT20}) and that has the desired value $x_i$ in the state register. This is possible because we have flag qubits in known locations. Because the entire process, without postselection, is unitary we can apply amplitude amplification. More precisely, by delaying all postselection we build the state
\begin{align*}
    \frac{1}{\sqrt{N}} \sum_{i=0}^{N-1}\ket{i} \ket{\hat{J}_0(x_i)}\ket{\hat\pi_0(x_i)} \left(\alpha \ket{\textnormal{good subspace}} + \beta \ket{\textnormal{bad subspace}}  \right) \, ,
\end{align*}
and we perform amplitude amplification on the state $\ket{i}\ket{\textnormal{anything}}\ket{\textnormal{good subspace}}$ to obtain the desired output in the given running time, remembering that $N=N_{\dr} N_{\ci}$ and $N_{\dr}\sim (T/\eps)^{d_{\dr}}$.
\end{proof}

\begin{remark}[Parameter $\gamma$] \label{rmk_gamma}
As mentioned in Theorem~\ref{thm_QDP} and Corollary~\ref{cor_eval_J0}, the parameter $\gamma$ defined as $\gamma:=\phi(0,T,L'_{g_{\mathrm{x}}},\mu_{g_{\mathrm{x}}},L'_{g_{\mathrm{u}}},\mu_{g_{\mathrm{u}}},L'_{J_T},\mu_{J_T})$ enters the running time. Of particular interest are problems where $\gamma=1$, because this avoids the exponential dependence on $T$ and $d$. Such problems include settings where 
\begin{enumerate}[(i)]
    \item $\kappa_{g_{\mathrm{x}}}=\kappa_{g_{\mathrm{u}}}=\kappa_{J_{T}}=1$, because in this case $\phi(0,T,L'_{g_{\mathrm{x}}},\mu_{g_{\mathrm{x}}},L'_{g_{\mathrm{u}}},\mu_{g_{\mathrm{u}}},L'_{J_T},\mu_{J_T})=1$;
    \item \label{it_ii_gamma} $\kappa_{g_{\mathrm{u}}}=\kappa_{J_{T}}=1$ and $g_{\mathrm{x}}$ is linear, because the LFT of a quadratic function with condition number $1$ is again a quadratic function with condition number $1$ (under the assumption that the primal and dual spaces are sufficiently large).
\end{enumerate}
We remark that the class of convex DPs with $\gamma=1$, although far from the full generality of DP, is already rich and difficult: Section~\ref{sec_hardness} discusses a DP problem with $\gamma=1$ that is $\#$P-hard.
\end{remark}

\subsection{Stochastic setting}
In this section we consider the stochastic setting as introduced in Section~\ref{sec_classical_stochastic}. Algorithm~\ref{algo_QDP_stoch} mentioned below solves the stochastic DP problem, where we recall that $\cM^d_{P}=\cQ^{d_{\dr}}_{P_{\dr}} \times \cR_{P_{\ci}}^{d_{\ci}}$. The details are given in the proof of Corollary~\ref{cor_QDP_stoch} which can be found in Appendix~\ref{app_pf_cor_QDP_stoch}. 

\begin{algorithm}[!htb]
\caption{Stochastic convex QDP}
\label{algo_QDP_stoch}
\begin{algorithmic}
\STATE \textbf{Input:} $T,r \in \N$, $(q_0,\ldots,q_{P_{\dr}-1})$,$(r_0,\ldots,r_{P_{\ci}-1})$, $(m_0,\ldots,m_{P-1})$, $(s_0,\ldots,s_{K_{\dr}})$, $(s'_0,\ldots,s'_{K_{\ci}})$, $a_i =-E^\top s_i$, $b_i =-C^\top s_i$, $d_i =-E^\top s'_i$, $o_{k,i}=A'(m_i+\xi_k)$, $\bar o_{k,i}=m_i+\xi_k$, and terminal value function $V'_T=:\hat V_T$; \vspace{2mm} \\
\item \textbf{Output:} Approximation to $\frac{1}{\sqrt{P}} \sum_{i=0}^{P-1}\ket{i} \ket{V_0(m_i)}\ket{\textnormal{Garbage}(i)}$; \vspace{2mm} \\
Prepare $\frac{1}{\sqrt{P}} \sum_{i=0}^{P-1} \ket{i} \ket{\hat{V}_T(m^{i+2T}_{i-2T})}$; \\
\item \textbf{For } $\ell=T,\ldots,1$ \textbf{ do} \\
\begin{enumerate}
\item \label{it_stoch1} $\frac{1}{\sqrt{K}} \sum\limits_{j=0}^{K_{\dr}-1}\sum\limits_{j'=0}^{K_{\ci}-1}\ket{j,j'} \ket{\hat{V}^*_{\ell}(s^{j+2\ell-1}_{j-2\ell+1},{s'}^{j'+2\ell-1}_{j'-2\ell+1})}^{\otimes r} \ket{\textnormal{Garbage}(j,j')}$;
\item\label{it_stoch2} $\frac{1}{\sqrt{K}} \!\sum\limits_{j=0}^{K_{\dr}-1}\!\sum\limits_{j'=0}^{K_{\ci}-1}\!\ket{j,j'}  \ket{\underbrace{\hat{V}^*_{\ell}(s^{j+2\ell-1}_{j-2\ell+1},{s'}^{j'+2\ell-1}_{j'-2\ell+1})\!+\!g_{\mathrm{u}}\hspace{-2.5mm}\protoast \!(a^{j+2\ell-1}_{j-2\ell+1},b^{j+2\ell-1}_{j-2\ell+1}+d^{j'+2\ell-1}_{j'-2\ell+1})}_{=:h_{\ell}(s^{j+2\ell-1}_{j-2\ell+1},{s'}^{j'+2\ell-1}_{j'-2\ell+1})} }^{\otimes r} \ket{\textnormal{Garbage}(j,j')}$;
\item \label{it_stoch3} $\frac{1}{\sqrt{P}}\!\! \sum\limits_{i=0}^{P-1}\!\ket{i} \ket{ h^*_{\ell}(o^{i+2(\ell-1)}_{0,i-2(\ell-1)})\!+\!g_{\mathrm{x}}(\bar o^{i+2(\ell-1)}_{0,i-2(\ell-1)})}\ldots \ket{ h^*_{\ell}(o^{i+2(\ell-1)}_{r-1,i-2(\ell-1)})\!+\!g_{\mathrm{x}}(\bar o^{i+2(\ell-1)}_{r-1,i-2(\ell-1)})} \ket{\textnormal{Garbage}(i)}$;
\item\label{it_stoch4} $\ket{\psi_{\ell}}=\frac{1}{\sqrt{P}} \sum\limits_{i=0}^{P-1} \ket{i}  \ket{\hat{V}_{\ell-1}(m_{i-2(\ell-1)}^{i+2(\ell-1)})} \ket{\textnormal{Garbage}(i)}$;
\end{enumerate}
\textbf{end}\\
Output $\ket{\psi_1}=\frac{1}{\sqrt{P}} \sum_{i=0}^{P-1} \ket{i}\ket{\hat{V}_0(m_{i})} \ket{\textnormal{Garbage}(i)}$;
\end{algorithmic}
\end{algorithm}

\begin{corollary}[Stochastic convex QDP] \label{cor_QDP_stoch}
Let $\eps>0, T \in \N$, $\cQ^{d_{\dr}}_{P_{\dr}}=\{q_0,\ldots,q_{P_{\dr}-1}\}$ be a regular discretization with $P_{\dr}\sim (T/\eps)^{d_{\dr}}$, and consider a stochastic discrete DP problem satisfying Assumptions~\ref{ass:modelling} and~\ref{ass:feasibility_stoch} where the shifted DP operator defined in~\eqref{eq_DP_shifted_stoch} preserves convexity.
Given Assumption~\ref{ass_precision}, the output $\ket{\psi_1}=\frac{1}{\sqrt{P}} \sum_{i=0}^{P-1}\ket{i} \ket{\hat{V}_0(m_i)}\ket{\textnormal{Garbage}(i)}$ for a successful run of Algorithm~\ref{algo_QDP_stoch} satisfies
\begin{align*}
|\hat{V}_0(m_i) - V_0(m_i)| \leq  \eps \quad \forall i \in [P]\,,
\end{align*}
where $\cQ_{P_{\dr}}^{d_{\dr}}\times \cR^{d_{\ci}}_{P_{\ci}}=\cM^d_{P}=\{m_0,\ldots,m_{P-1}\}$.
Given Assumption~\ref{ass_Uf}, the expected running time for constant probability of success is 
\begin{align*}
O\Big(r T \gamma^{drT} \, \polylog\big(P,(T/\eps)^{d}\big)\Big) \, ,
\end{align*}
with \smash{$\gamma:=\phi(0,T,L'_{g_{\mathrm{x}}},\mu_{g_{\mathrm{x}}},L'_{g_{\mathrm{u}}},\mu_{g_{\mathrm{u}}},L'_{V_T},\mu_{V_T})$}, where the function $\phi$ is defined in Lemma~\ref{lem_condNumber_valueFunction}.
Combined with amplitude amplification Algorithm~\ref{algo_QDP_stoch} outputs $\hat{V}_0(m_i)$ such that $|\hat{V}_0(m_i) - V_0(m_i)| \leq \eps$ for any $i \in [P]$ with an expected running time 
\begin{align*}
O\Big(r T \gamma^{drT} (T/\eps)^{d_{\dr}/2}\sqrt{P_{\ci}}\, \polylog\big(P_{\ci},(T/\eps)^{d}\big)\Big) \, .   
\end{align*}
Furthermore, Algorithm~\ref{algo_QDP_stoch} can compute $\hat \pi_0(m_i)$ such that $\norm{\hat \pi_0(m_i)-\pi_0^{\star}(m_i)} \leq \sqrt{4\eps/\mu_{g_{\mathrm{u}}}}$, where $\pi_0^{\star}(m_i)$ denotes the optimal initial action for state $m_i$, at the additional cost of solving the convex optimization problem~\eqref{eq_optPolicy_via_LFT_robust_stoch}.
\end{corollary}

\section{Discussion on generality, running time, and optimality}
In this section we show that 
\begin{enumerate}[(i)]
    \item the class of DP problems we can solve with the quantum algorithm presented in Section~\ref{sec_quantum_algo} contains $\#$P-hard problems (see Section~\ref{sec_hardness});
    \item for quadratic continuous stochastic DP problems, the quantum algorithm achieves a quadratic speedup compared to the standard classical approach using discretization (see Section~\ref{sec_speedup}). 
\end{enumerate}
It is widely believed that quantum computers cannot achieve more than a quadratic speedup on $\#$P-hard problems~\cite{bennett97,BV97} (although we note that for the computation of the Jones polynomial, which is a $\#$P-hard problem, there exists an efficient quantum {\em approximation} algorithm even if no classical polynomial-time approximation algorithm is known \cite{aharonov2009polynomial}). Since our algorithm achieves a quadratic speedup for some DP problems, it may be optimal at least for those problems. However, we are not aware of a classical lower bound matching the running time of discretization followed by the standard Bellman recursion; in other words, the classical algorithm that we compare to may not be optimal.
In Section~\ref{sec_oracle_optimal}, we show that the quantum algorithm is optimal in the oracle setting, up to polylogarithmic factors.
\subsection{Framework contains \#P-hard problems} \label{sec_hardness}
Let $Z_1,\dots,Z_n$ be a sequence of of random variables with support $\{\alpha_{i,j} \in \N:i\in [n], j\in\{1,2\} \}$, where $\alpha_{i,2}=0$ for all $i \in [n]$, and probabilities $\PP(Z_i = \alpha_{i,j})=1/2$ for all $i\in [n]$, $j\in\{1,2\}$.
For $T=n+2$ consider the following one-dimensional DP problem 
  \begin{align}  \label{eq_hardDPproblem}
  \begin{array}{cl}
    \min\limits_{\pi_0,\ldots,\pi_{T-1}} & \mathbb{E}\left[ \sum_{t=0}^{T-1} \big(g_{\mathrm{x},t}(x_t) + g_{\mathrm{u},t}(u_t)\big)  + g_T(x_T)\right] \\
    \textnormal{s.t.} & x_{t+1} = a_t x_t + b_t u_t + \xi_{t+1}, \, x_0=0 \\
    &x_t \in \X, \ t=0,\hdots, T \\
    &u_t=\pi_t(x_t) \in \U,\ t=0,\hdots, T-1 \, ,
    \end{array}
  \end{align}
for $\X=[-U_{\mathrm{x}},U_{\mathrm{x}}]$ and $\U=[0,U_{\mathrm{u}}]$, where $g_T$ and $g_{\mathrm{u},t}$ are quadratic functions (with nonzero lead coefficient), and $g_{\mathrm{x},t}=0$ for all $t=0,\ldots,T-1$. In addition, $a_t=b_t=1$ for all $t=1,\ldots,T-2$ and $a_{T-1}=-b_{T-1}=1$.
Furthermore, we have $\xi_1 = \xi_{T} = 0, \xi_{t+1} = -Z_t \; \forall t = 1,\dots,T-2$. This is a one-dimensional continuous stochastic convex DP problem with specific structure; we show that this problem is already hard, so that, by extension, the class of convex DP problems studied in this paper is hard as well.

\begin{proposition} \label{prop_hardDP}
It is \#P-hard to compute the optimal initial action $\pi^\star(x_0)$ of problem \eqref{eq_hardDPproblem}.
\end{proposition} 
The proof is given in Appendix~\ref{app_pf_hardness}.
Problem~\eqref{eq_hardDPproblem} can be solved by Algorithm~\ref{algo_QDP_stoch} because it satisfies Assumptions~\ref{ass:modelling} and~\ref{ass:feasibility_stoch}, and it is a convex DP problem, i.e., the shifted DP operator is convexity preserving, as ensured by Lemma~\ref{lem_cont_convexity_preserving}. We also remark  that for this problem, the running time parameter $\gamma$  is equal to $1$, as ensured by Remark~\ref{rmk_gamma}.

\subsection{Quadratic quantum speedup} \label{sec_speedup}
In this section we show that Algorithm~\ref{algo_QDP_stoch} achieves a quadratic speedup compared to the classical Bellman approach for continuous stochastic DP problems.
\paragraph{One-dimensional problems.}
For continuous stochastic DP problems with a one-dimensional state and action space and quadratic cost functions, Algorithm~\ref{algo_QDP_stoch} computes an $\eps$-approximation of the value function $V_0(m_i)$ and a $\sqrt{4\eps/\mu_{g_{\mathrm{u}}}}$-approximation of the corresponding optimal policy $\pi^\star_0(m_i)$ in time
\begin{align} \label{eq_quantum_runtime_1d}
    \tilde O\big(r T^{3/2}/\sqrt{\eps}\big) \, , 
\end{align}
see Corollary~\ref{cor_QDP_stoch}.\footnote{Recall that $\gamma=1$ for these problems, as discussed in Remark~\ref{rmk_gamma}} The optimization problem \eqref{eq_optPolicy_via_LFT_robust_stoch}, required to solve for obtaining the optimal policy, does not affect the overall running time (modulo polylogarithmic factors) because it can be done efficiently via binary search.

The standard classical approach to tackle such problems is to discretize the state and action space, and then use the textbook DP algorithm to calculate an $\eps$-solution to the value function. The running time thus scales as $O(r T |\cX_N| \polylog(|\cU_M|)$, where $\cX_N$ and $\cU_M$ denote the discretized state and action space, respectively. Because the problem is one-dimensional, the optimization step over the action space can be solved with binary search, hence $M$ appears polylogarithmically in the running time expression. To ensure an $\eps$-approximation of the value function, the discretization parameter $N$ needs to be of order $N\sim T/\eps$; thus, the overall running time scales as
\begin{align} \label{eq_classical_runtime_1d}
    \tilde O\big(r T^{2}/\eps\big) \, .
\end{align}
Given an approximation with error $\eps$ of the value function $V_0(m_i)$, we can can approximately compute the corresponding optimal policy $\pi^\star_0(m_i)$, up to  error $\sqrt{4\eps/\mu_{g_{\mathrm{u}}}}$, using a similar argument as in the proof of Corollary~\ref{cor_optimal_policy_stoch} (the proof is based on the LFT approach, but the same error bound can be proven for the standard Bellman recursion on a discretized state space).

Comparing~\eqref{eq_quantum_runtime_1d} to~\eqref{eq_classical_runtime_1d} shows that we obtain a quadratic quantum speedup in $N=T/\eps$. We remark that there may be specialized classical approaches that are potentially more efficient than the standard Bellman recursion on a discretized problem. For example, \cite{halman2016fully} describes a fully polynomial-time approximation scheme (FPTAS) for one-dimensional continuous stochastic convex DP problems; the assumptions of their model are slightly different (e.g., the cost functions do not have to be strongly convex, but they have to be nonnegative, and the approximation scheme returns a solution with absolute and relative error, both of which have to be nonzero), and the running time of the algorithm is $\tilde{O}(T^2/\varepsilon)$, where $\varepsilon$ is the relative error --- differently from the absolute error used everywhere else in this paper.

\paragraph{Multidimensional problems.} For continuous stochastic DP problems with $d$-dimensional state space, and quadratic cost functions with condition number $1$, Algorithm~\ref{algo_QDP_stoch} computes an $\eps$-approximation of $V_0(m_i)$ in time
\begin{align}\label{eq_quantum_runtime_d}
    \tilde O\big(r T (T/\eps)^{d/2} \big) \, ,
\end{align}
as guaranteed by Corollary~\ref{cor_QDP_stoch}. Algorithm~\ref{algo_QDP_stoch} also computes a $\sqrt{4\eps/\mu_{g_{\mathrm{u}}}}$-approximation of the corresponding optimal policy $\pi^\star_0(m_i)$; this requires solving the optimization problem \eqref{eq_optPolicy_via_LFT_robust_stoch}, which, since the original problem is continuous, can be solved as a convex mathematical optimization problem in time $\tilde{O}(\poly(d))$.

The classical Bellman recursion, with discretization step $\sim T/\epsilon$ along each axis as in the quantum algorithm, computes an $\eps$-approximation of $V_0(m_i)$ and a $\sqrt{4\eps/\mu_{g_{\mathrm{u}}}}$-approximation of $\pi^\star_0(m_i)$ in time
\begin{align}\label{eq_classicl_runtime_d}
    \tilde O\big(r T (T/\eps)^{d} T_{\mathrm{act}} \big) \, ,
\end{align}
where $T_{\mathrm{act}}$ is the time to choose the optimal action for a given state. As in the one-dimensional case, we observe a quadratic quantum speedup in $N \sim T/\epsilon$ when comparing~\eqref{eq_quantum_runtime_d} with~\eqref{eq_classicl_runtime_d}.

\subsection{Optimality in the oracle setting} \label{sec_oracle_optimal}
The purpose of this section is to show that the quantum algorithm for solving deterministic DP problems presented in Section~\ref{sec_quantum_algo} cannot be substantially improved for problems with $\gamma = 1$.
\begin{proposition} \label{prop_lower_bound_det}
There exist deterministic DP problems satisfying Assumptions~\ref{ass:modelling} and~\ref{ass:feasibility} with $\cX_{2^d}^d \subset [0,1]^d$ such that any quantum algorithm that outputs $J_0(x)$ for any $x\in \cX_{2^d}^d$ requires $\Omega(\sqrt{2^d}/d)$ evaluations of the cost functions.
\end{proposition}
\begin{proof}
We consider a sequence of purely discrete DP problems where $T=1$, and the discretized state space is chosen to be $\{0,1\}^d$. For some $\alpha \in \{0,1\}^d$, we choose the terminal cost function (and also value function) to be $J_1(x)=\max_{i \in [d]}|x_i - \alpha_i|$, similar to~\cite[Proposition~6.1]{QLFT20}; it is easy to see that a call to $J_1(x)$ can be simulated with one call to a function $f_\alpha$ such that $f_\alpha(y)=1$ if and only if $y=\alpha$. Furthermore let $g_{\mathrm{x}}(\cdot)=g_{\mathrm{u}}(\cdot)=0$. For every $k=1,\dots,d$, construct a problem where $A'$ is the all-zero matrix and $B'=(e_1,\ldots,e_{k-1},e_{k+1},\ldots,e_d)$, where $e_k$ is the all-zero vector with a one at the $k$-th entry. We thus see that $J_0(x)=|x_k-\alpha_k|$ for all $k \in [d]$; indeed, at state $x$ the action allows us to change all the digits except position $k$, so if $\alpha_k = x_k$ we can choose action $\alpha-x$ and pay total cost $0$, if $\alpha_k \neq x_k$ we cannot reach the optimal $\alpha$ so we have to pay $1$. As a result, by evaluating $J_0(e_k)$ we are able to determine $\alpha_k$, and by repeating this process $k$ times we can determine $\alpha$, evaluating the value function at $d$ different points. It is known that determining $\alpha$ requires $\Omega(\sqrt{2^d})$~\cite{zalka1999grover} evaluations of $f_\alpha$, hence evaluating value function at a single point requires $\Omega(\sqrt{2^d}/d)$ cost function evaluations.
\end{proof}
Proposition~\ref{prop_lower_bound_det} establishes a lower bound on the number of evaluations of the cost functions, and Algorithm~\ref{algo_QDP} attains this lower bound from Proposition~\ref{prop_lower_bound_det} up to polylogarithmic factors. To see this, note that for the problem described in the proof above we have $\gamma=1$: this can be verified by evaluating the $\phi$ function and working with the $W$-parameter~\eqref{eq_W} introduced in~\cite{QLFT20}, instead of the condition number if the function is not differentiable; see a similar discussion in \cite[Section 6]{QLFT20}. Furthermore, \cite[Section 6]{QLFT20} shows that there is an appropriate choice of primal and dual space for the function $J_1$, so that the QLFT computations are exact, implying that we are able to compute $J_0$ with no error. Finally, note that the value function at each stage is convex, as shown in the proof. Hence, by Corollary~\ref{cor_eval_J0} Algorithm~\ref{algo_QDP} outputs value function at a specific point in time $O(\sqrt{2^d}\, \polylog(2^d,(\sqrt{d}/\eps)^d))$: up to polylogarithmic factors, this matches the lower bound of  Proposition~\ref{prop_lower_bound_det}.




\paragraph{Acknowledgments}
We thank Peyman Mohajerin Esfahani for discussions on the connections between dynamic programming and the discrete Legendre-Fenchel transform, related to~\cite{peyman}.

\appendix
\section{Proofs}
\subsection{Proof of Lemma~\ref{lem_discreteLFT_approx}} \label{app_pfLemmadiscLFT}
We start by proving the first statement.
By definition of the LFT we can write:
\begin{align*}
    |f\!\protoast\!(s) - f^*(s)|
    = \max_{x \in \X^d}\{ \inprod{s}{x}-f(x)\} - \max_{z \in \cX_N^d}\{ \inprod{s}{z} -f(z)\}
    \leq \inprod{s}{x^\star-z'} + |f(x^\star) - f(z')| \, ,
\end{align*}
where $x^\star$ denotes the optimizer of the first maximization problem, and $z'$ is any point in $\cX_N^d$; in particular, we can choose $z'$ to be the point in $\cX_N^d$ closest to $x^\star$. Applying the Cauchy-Schwarz inequality then gives
\begin{align*}
   |f\!\protoast(s) - f^*(s)|
   \leq (\norm{s} + L_f) \norm{x^\star-z'}
   \leq (\sqrt{d}L_f+L_f) \distH(\X^d,\cX_N^d) \, ,
\end{align*}
where the final step uses the fact that the largest element in the dual space is bounded by $L_f$~\cite[Remark~3.2]{QLFT20}.

We next prove the second statement of the lemma. By definition of the LFT, for any $s_1,s_2 \in \mathbb{S}^d$ we have
\begin{align*}
|f\!\protoast(s_1) -f\!\protoast(s_2)|
=\big| \max_{x \in \X^d}\{\inprod{s_1}{x} - f(x)\} -\max_{x \in \X^d}\{\inprod{s_2}{x} - f(x)\}  \big| \, .
\end{align*}
Let $x_1^\star$ and $x_2^\star$ denote the optimizers of the first and second maximization problem above. We then define $x^\star=x_1^\star$ if $\inprod{s_1}{x_1^\star} - f(x_1^\star) \geq \inprod{s_2}{x_2^\star} - f(x_2^\star)$ and $x^\star=x_2^\star$, otherwise. Hence, we find
\begin{align*}
|f\!\protoast(s_1) -f\!\protoast(s_2)|
\leq |\inprod{s_1-s_2}{x^\star}|
\leq \norm{s_1-s_2} \Delta_{\X^d} \, ,
\end{align*}
where the last step uses Cauchy-Schwarz. The statement for the discrete LFT follows analogously. 

We next prove the third statement. By the triangle inequality we find
\begin{align*}
|f\!\protoast\!\protoast(x) - f^{**}(x)|
&\leq |f\!\protoast\!\protoast(x) - f\!\protoast^*(x)|+|f\!\protoast^*(x) - f^{**}(x)|\\
&\leq (1+\sqrt{d})L_{f\protoast} \distH(\mathbb{S}^d,\cS_K^d) + (1+\sqrt{d})L_{f} \distH(\mathbb{X}^d,\cX_N^d) \\
& \leq (1+\sqrt{d})\Delta_{\X^d} \distH(\mathbb{S}^d,\cS_K^d) + (1+\sqrt{d})L_{f} \distH(\mathbb{X}^d,\cX_N^d) \, ,
\end{align*}
where the penultimate step uses Lemma~\ref{lem_contLFT} and~\eqref{it_LFT1}. The final step follows from~\eqref{it_LFT2}.

It remains to prove the final statement of the lemma.
In case $\X^d$ is compact convex the statement follows from the  Fenchel-Moreau theorem~\cite{rockafellar70}. If $\X^d$ contains discrete parts the statement requires some more work. Let $\bar f$ denote the convex extension of $f$ such that $f(x)=\bar f(x)$ for all $x \in \X^d$.
By definition of the LFT we have
\begin{align} \label{eq_discLFT1}
f\!\protoast\!\protoast(x)
=\max_{s \in \mathbb{S}^d} \min_{x' \in \X^d} \{ \inprod{x-x'}{s}+f(x') \} \, .
\end{align}
We thus see for all $x \in \X^d$
\begin{align} 
f\!\protoast\!\protoast(x)
&= \max_{s \in \mathbb{S}^d} \min_{x' \in \X^d} \{ \inprod{x-x'}{s}+\bar f(x') \} \nonumber  \\
&\geq \max_{s \in \mathbb{S}^d} \min_{x' \in \conv(\X^d)} \{ \inprod{x-x'}{s}+\bar f(x') \} 
= \bar f\!\protoast\!\protoast(x)
=\bar f(x)
=f(x) \, , \label{eq_discLFT2}
\end{align}
where the penultimate step follows from the Fenchel-Moreau theorem.
By choosing $x'=x$ in~\eqref{eq_discLFT1} we find $f\!\protoast\!\protoast(x) \leq f(x)$ which together with~\eqref{eq_discLFT2} proves the assertion.
\qed
\subsection{Proof of Proposition~\ref{prop_QLFT}} \label{app_proof_QLFT}
From the starting state we prepare 
\begin{align}
   \frac{1}{\sqrt{N}} \sum_{i=0}^{N-1} \ket{i} \ket{x_{i-\ell -1}^{i+\ell+1}}  \ket{f(x_{i-\ell -1}^{i+\ell+1})} \ket{c_{i-\ell-1}^{i+\ell}} \ket{\textnormal{Garbage}(i)} \, , \label{eq_step1_dsss}
\end{align}
following the same steps as in~\cite[Step~2 in the proof of Theorem~4.7]{QLFT20}. We next define a parameter 
\begin{align} 
W:=\left \lfloor \max_{i \in\{1,\ldots,N-2\}}\{c_{\ci}-c_{i-1}\} \frac{1}{\delta_s}  \right \rfloor \, , \label{eq_W}
\end{align}
a set
\begin{align*} 
    &\cA:=\Big \lbrace (i,m) \in [N]\times [W]: \Big(\left \lfloor \frac{c_{\ci} - c_{i-1}}{\delta_s} \right \rfloor \geq  m+1 \wedge i\in \{1,\ldots,N-2\}\Big)  \nonumber \\
    &\hspace{60mm}\vee (m=0 \wedge i=0)\vee (m=0 \wedge i=N-1) \Big \rbrace\, ,
\end{align*}
and a function
\begin{align*}
 j(i,m,c_{i-1}):=\left \lbrace
 \begin{array}{cl}
    \emptyset  & \textnormal{if } (i,m) \not \in \cA\\
      0 &  \textnormal{if } i=0 \wedge m=0 \\
      k-1 &\textnormal{if } i=n-1 \wedge m=0\\
      \min\limits_{\ell}\{\ell+m:\, c_{i-1} < s_{\ell} \wedge \ell\in[N] \} & \textnormal{otherwise}\,.
 \end{array}
 \right. 
\end{align*}
The intuition for these parameters is given in~\cite[Section~3.1]{QLFT20}, see also the proof of \cite[Theorem 4.7]{QLFT20}.
We next evolve the state~\eqref{eq_step1_dsss} into
\begin{align*}
 &\frac{1}{\sqrt{NW}} \sum_{i=0}^{N-1} \sum_{m=0}^{W-1} \ket{i}\ket{x_{i-\ell -1}^{i+\ell+1}}  \ket{f(x_{i-\ell -1}^{i+\ell+1})} \ket{c_{i-\ell-1}^{i+\ell}} \ket{m} \ket{\mathds{1}\{(h,m) \in \cA \, \forall \, h=i-\ell-1,\ldots,i+\ell+1 \}} \nonumber \\ 
 &\hspace{80mm}\ket{j(i,m,c_{i-1})} \ket{\textnormal{Garbage}(i)} \, . 
\end{align*}
We then uncompute the registers $\ket{x_{i-\ell-1},x_{i+\ell+1}}\ket{f(x_{i-\ell-1}),f(x_{i+\ell+1})}\ket{c_{i-\ell-1}^{i+\ell}}$, which gives
\begin{align}
&\frac{1}{\sqrt{NW}} \sum_{i=0}^{N-1} \sum_{m=0}^{W-1} \ket{i} \ket{x_{i-\ell}^{i+\ell}} \ket{f(x_{i-\ell}^{i+\ell})} \ket{m} \ket{\mathds{1}\{(h,m) \in \cA \, \forall \, h=i-\ell-1,\ldots,i+\ell+1 \}} \nonumber \\
&\hspace{80mm}\ket{j(i,m,c_{i-1})} \ket{\textnormal{Garbage}(i)} \, . \label{eq_middle}
\end{align}
The algorithm is successful when the indicator function has value ``$1$'', as will be seen below. Before we discuss that, we observe that the indicator function can indeed have the value 1, i.e., the condition $(h,m) \in \cA$ for all $h=i-\ell-1,\ldots,i+\ell+1$ is verified for some value of $i$ and $m$. Indeed, by convexity of the function $f$, if the condition $ \lfloor \frac{c_{h} - c_{h-1}}{\delta_s} \rfloor \geq  m+1$ in the definition of $\cA$ is verified for $h=i-\ell-1$, it is verified for all other $h=i-\ell,\ldots,i+\ell+1$. As discussed in \cite{QLFT20}, the condition must be verified for some choice of $i$ and $m$, which shows that the indicator must have value 1 for some choice of $i$ and $m$. If we perform a measurement on the register with the indicator function, then conditioned on seeing the outcome ``$1$" and after a relabelling of the sum we obtain
\begin{align*}
   \frac{1}{\sqrt{K}} \sum_{j=0}^{K-1} \ket{j} \ket{\bar{x}_{j-\ell}^{j+\ell}}  \ket{f(\bar{x}_{j-\ell}^{j+\ell})}\ket{\textnormal{Garbage}(j)} \, ,  
\end{align*}
where $\bar{x}_j$ denotes the optimizer given in the definition of the LFT, i.e, $f^*(s_j)=s_j \bar{x}-f(\bar x)$. This step is probabilistic, and succeeds with probability $K/(NW) \geq 1/\kappa_f$~\cite[Theorem~4.5]{QLFT20}, where we used the fact that the indicator function in~\eqref{eq_middle} maps $N\times W$ nonzero indices to $K$ nonzero indices, due to the fact that each point in the dual space must have an optimizer in the primal space.
The final step follows the same lines as~\cite[Step~4 in proof of Theorem~4.5]{QLFT20}.\qed

\subsection{Proof of Theorem~\ref{thm_classical_error}} \label{app_pf_thm_error_bound}
For a compact convex space $\mathbb{\tilde S}^d\subseteq \R^d$ let $\mathbb{S}^d:=\{A'^\top s : s \in \mathbb{\tilde S}^d \}\subseteq \R^d$.
For the proof of Theorem~\ref{thm_classical_error} we use the notation $x=(y,z)$, and define the continuous conjugate DP operator analogously to $\mathrm{DP_{\!conj}}[J']$, with the difference that all the LFTs are continuous rather than discrete, i.e.,
\begin{align}
\mathrm{DP_{\!conj}^{cont}}[J](x) 
    :=\big(J\!\protoast\!(s) + g_{\mathrm{u}}\!\!\!\protoast\!(-{B'}^\top s) \big)\!\protoast\!(A'x) + g_{\mathrm{x}}(x) \, ,
\end{align}
for primal and dual spaces $\X^d$ and $\mathbb{\tilde S}^d$, respectively.
We next relate this operator with the biconjugate $\mathrm{DP_{shift}}[J]\!\protoast\!\protoast(x)$ for primal and dual spaces $\X^d$ and $\mathbb{S}^d$, respectively.
\begin{lemma} \label{lem_DP_conj_cont}
For the setting of Theorem~\ref{thm_classical_error} we have $\mathrm{DP_{\!conj}^{cont}}[J](x)=\mathrm{DP_{shift}}[J]\!\protoast\!\protoast(x) + g_{\mathrm{x}}(x)$ $\forall x \in \X^d$.
\end{lemma}
\begin{proof}
By definition continuous conjugate DP operator we have
\begin{align*}
\mathrm{DP_{\!conj}^{cont}}[J](x)-g_{\mathrm{x}}(x)
&= \big(J\!\protoast(s)+g_{\mathrm{u}}\!\!\!\protoast(-B^\top s) \big)\!\protoast(A'x) \\
&= \sup_{s\in \mathbb{\tilde S}^d} \big\{ \inprod{A'x}{s} - J\!\protoast(s) - g_{\mathrm{u}}\!\!\!\protoast(-B^\top s) \big\}\\
&=\sup_{s\in \mathbb{\tilde S}^d} \min_{x' \in \X^d,u\in \U^c} \big\{ \inprod{Ax+Bu-x'}{s} +J(x') + g_{\mathrm{u}}(u) \big\} \, .
\end{align*}
Substituting $x'=A'\tilde x+B'u \in \X^d$ for $\tilde x \in \X^d$ gives
\begin{align}
\mathrm{DP_{\!conj}^{cont}}[J](x)-g_{\mathrm{x}}(x)
&=\max_{s\in \mathbb{\tilde S}^d} \min_{\tilde x \in \X^d,u\in \U^c} \big\{ \inprod{x-\tilde x}{A'^\top s} +J(A'\tilde x+B'u) + g_{\mathrm{u}}(u) \big\}\nonumber\\
& =\max_{s\in \mathbb{ S}^d} \min_{\tilde x \in \X^d,u\in \U^c} \big\{ \inprod{x-\tilde x}{s} +J(A'\tilde x+B'u) + g_{\mathrm{u}}(u) \big\}\label{eq_dsLemma_1} \, .
\end{align}
By definition of the shifted DP operator~\eqref{eq_DP_tilde} we find
\begin{align*}
 \mathrm{DP_{shift}}[J]\!\protoast\!\protoast(x)
 &= \max_{s\in \mathbb{ S}^d} \big\{\inprod{x}{s} -\mathrm{DP_{shift}}[J]\!\protoast(s) \big\} \\
 &= \max_{s\in \mathbb{ S}^d} \min_{\tilde x\in \X^d, u \in \U^c} \big\{\inprod{x-\tilde x}{s} +J(A'\tilde x+B'u) + g_{\mathrm{u}}(u) \big\} \, ,
\end{align*}
 which together with~\eqref{eq_dsLemma_1} proves the assertion. 

\end{proof}

\begin{proof}[Proof of Theorem~\ref{thm_classical_error}]
Recall that $x=(y,z)$, $\X^d=\Y^{d_{\dr}}\times \cZ^{d_{\ci}}_{N_{\ci}}$, and $\cX^d_N=\cY_{N_{\dr}}^{d_{\dr}} \times \cZ^{d_{\ci}}_{N_{\ci}}$.
We start by proving that $\mathrm{DP_{\!conj}}[J']$ is convex extensible. To do so recall that by definition of the conjugate DP operator~\eqref{eq_conjDPoperator} we have for all $y\in \cY_{N_{\dr}}^{d_{\dr}}$ and $z \in \cZ^{d_{\ci}}_{N_{\ci}}$
\begin{align}
 &\mathrm{DP_{conj}}[J'] (y,z)\nonumber \\
 &= g_{\mathrm{x}}(y,z) + h^*(Ay,Dz) \nonumber \\
 &= g_{\mathrm{x}}(y,z) + \max_{s\in\cS^{d_{\dr}}_{K_{\dr}}, s'\in {\cS'}^{d_{\ci}}_{K_{\ci}}}\big\{\inprod{s}{Ay}+\inprod{s'}{Dz} - h(s,s')\big\} \nonumber \\
 &= g_{\mathrm{x}}(y,z) + \max_{s\in\cS^{d_{\dr}}_{K_{\dr}}, s'\in {\cS'}^{d_{\ci}}_{K_{\ci}}}\{\inprod{s}{Ay}+\inprod{s'}{Dz} - g_{\mathrm{u}}\hspace{-2.5mm}\protoast \!(-B^\top s,-C^\top s -E^\top s') - {J'}^*(s,s')\}\nonumber  \\
 &=  g_{\mathrm{x}}(y,z) + \max_{s\in\cS^{d_{\dr}}_{K_{\dr}}, s'\in {\cS'}^{d_{\ci}}_{K_{\ci}}}\Big\{\inprod{s}{Ay}+\inprod{s'}{Dz}- \nonumber\\ 
 &\hspace{3mm}\max_{v\in\V^{d_{\dr}},w\in \cW^{c_{\ci}}_{M_{\ci}}}\!\!\!\big\{\inprod{-B^\top s}{v}\!+\!\inprod{-C^\top s-E^\top s'}{w}-g_{\mathrm{u}}(v,w)\big \}\!-\!\!\!\!\!\!\!\! \max_{q\in\cY^{d_{\dr}}_{N_{\dr}},q'\in \cZ^{d_{\ci}}_{N_{\ci}}}\!\!\!\!\big\{\inprod{s}{q}+\inprod{s'}{q'} - {J'}(q,q')\big\}\Big\} \nonumber \\
 &=  g_{\mathrm{x}}(y,z) + \max_{s\in\cS^{d_{\dr}}_{K_{\dr}}, s'\in {\cS'}^{d_{\ci}}_{K_{\ci}}}\Big\{\inprod{s}{Ay}+\inprod{s'}{Dz} +  \nonumber\\ 
 &\hspace{3mm}\min_{v\in\V^{d_{\dr}},w\in \cW^{c_{\ci}}_{M_{\ci}}}\!\!\!\big\{\inprod{B^\top s}{v}\!+\!\inprod{C^\top s\!+\!E^\top s'}{w}\!+\!g_{\mathrm{u}}(v,w)\big \}\!+\!\!\!\!\!\!\!\! \min_{q\in\cY^{d_{\dr}}_{N_{\dr}},q'\in \cZ^{d_{\ci}}_{N_{\ci}}}\!\!\big\{{J'}(q,q')\!-\!\inprod{s}{q}\!-\!\inprod{s'}{q'} \big\}\Big\}\, . \label{eq_convexExt1}
\end{align}
To show that $\mathrm{DP_{conj}}[J'](y,z)$ is convex extensible it suffices to show that there exists a convex function defined on the convex hull of $\cY^{d_{\dr}}_{N_\dr}\times \cZ^{d_\ci}_{N_\ci}$ that matches $\mathrm{DP_{conj}}[J'] (y,z)$ on the discrete points. This however follows from~\eqref{eq_convexExt1}: by Assumption~\ref{ass:modelling}, the function $g_{\mathrm{x}}$ is convex extensible; furthermore, the pointwise maximum of a family of convex functions (in $y, z$) is convex, which proves the assertion~\eqref{it_convexExtensible}.

We next prove the Statement~\eqref{it_Lipschitz}. By definition of the conjugate DP operator~\eqref{eq_conjDPoperator} we have for all $x,x' \in \cX_N^d$
\begin{align*}
|\mathrm{DP_{conj}}[J'](x) - \mathrm{DP_{conj}}[J'](x')|
& \leq |h^*(A'x)-h^*(A'x')| + |g_{\mathrm{x}}(x)-g_{\mathrm{x}}(x')| \\
& \leq \sqrt{d} L_{J} \| A'x - A'x' \| + L_{g_{\mathrm{x}}}\norm{x-x'} \\
& \leq (\sqrt{d} L_{J}\| A' \|_{\infty} + L_{g_{\mathrm{x}}}) \norm{x-x'} \\
& = (\sqrt{d} L_{J}\max\{\norm{A}_{\infty},\norm{D}_{\infty} \} + L_{g_{\mathrm{x}}}) \norm{x-x'} \, ,
\end{align*}
where the second step uses Statement~\eqref{it_LFT2} of Lemma~\ref{lem_discreteLFT_approx} and that nontrivial dual space for the discrete LFT if bounded by the Lipschitz constant of the function to be transformed~\cite[Remark~3.2]{QLFT20}, which implies $\Delta_{\cS^d_K} \leq \sqrt{d}L_J$. 
The penultimate step above follows by definition of the operator norm. This completes the proof of~\eqref{it_Lipschitz}.

It remains to prove Statement~\eqref{it_errorBound}. 
Lemma~\ref{lem_DP_conj_cont} implies that for all $x \in \X^d$
\begin{align}
  \mathrm{DP_{shift}}[J]\!\protoast\!\protoast(x)+ g_{\mathrm{x}}(x)
  &= \mathrm{DP_{\!conj}^{cont}}[J](x)\nonumber\\
  &= \big(J\!\protoast\!(s) + g_{\mathrm{u}}\!\!\!\protoast\!(-{B'}^\top s) \big)\!\protoast\!(A'x) \nonumber\\
  &\geq \big(J^*(s) + g_{\mathrm{u}}\!\!\!\protoast\!(-{B'}^\top s) + (1+\sqrt{d})L_{J}\distH(\X^d,\cX_N^d) \big)\!\protoast\!(A'x) \nonumber \\
  &=\big(J^*(s) + g_{\mathrm{u}}\!\!\!\protoast\!(-{B'}^\top s)\big)\!\protoast\!(A'x)  - (1+\sqrt{d})L_{J}\distH(\X^d,\cX_N^d) \, , \label{eq_DSS1}
\end{align}
where the penultimate step uses Lemma~\ref{lem_discreteLFT_approx}, the inequality $J\protoast \le J^* + (1+\sqrt{d})L_{J}\distH(\X^d,\cX_N^d)$, and Property~\eqref{it_reverse} from Fact~\ref{fact_LFT_basic}. The final step follows from Statement~\eqref{it_shift} in Fact~\ref{fact_LFT_basic}. Defining $\omega(s):=J^*(s) + g_{\mathrm{u}}\!\!\!\protoast\!(-{B'}^\top s)$ and using Lemma~\ref{lem_discreteLFT_approx} gives
\begin{align}
\omega\!\protoast\!(A'x)
&\geq \omega^*(A'x) - (1+\sqrt{d})L_{\omega}\distH(\mathbb{S}^d,\cS_K^d) \, , \label{eq_DSS2}
\end{align}
where $L_{\omega} \leq L_{J^*} + L_{g_{\mathrm{u}}\!\!\!\protoast} \leq \tau + \eta$. To see the final step note that 
\begin{align*}
    |J^*(s)-J^*(s')| 
    \leq |\max_{x\in \cX_N^d}\{\inprod{s}{x}-J(x)\} -\max_{x\in \cX_N^d}\{\inprod{s'}{x}-J(x)\} |
    \leq |\inprod{s-s'}{x^\star}|
    \leq \norm{s-s'} \norm{x^\star} \, .
\end{align*}
Combining~\eqref{eq_DSS1} and~\eqref{eq_DSS2} implies
\begin{align*}
 \mathrm{DP_{shift}}[J]\!\protoast\!\protoast(x)+ g_{\mathrm{x}}(x)
 \geq \big(J^*(s) + g_{\mathrm{u}}\!\!\!\protoast\!(-{B'}^\top s)\big)^*(A'x)  - E_1 -E_2
 =\mathrm{DP_{\!conj}}[J'](x) - E_1 - E_2 \, .
\end{align*}
The same steps can be applied to verify $\mathrm{DP_{shift}}[J]\!\protoast\!\protoast(x)+ g_{\mathrm{x}}(x) \leq \mathrm{DP_{\!conj}}[J'](x)  + E_1 + E_2$ which completes the proof.
\end{proof}

\subsection{Proof of Lemma~\ref{lem_optimal_policy}}\label{app_pf_optimal_policy}
By assumption the parameters $N_{\dr}$ and $K$ be sufficiently large such that
\begin{align} \label{eq_byass}
     E_1 + E_2 \leq \eps \, ,
\end{align}
For $E_1$ and $E_2$ defined in Theorem~\ref{thm_classical_error}.
Assume for the sake of contradiction that $\pi^\star(x)$ and $\hat \pi(x)$ are such that $\|\pi^\star(x)-\hat \pi(x)\| > \sqrt{4\eps/\mu_{g_{\mathrm{u}}}}$. From the DP scheme via the LFT we recall that
\begin{align*} 
 h^{*}(A'x)
=\inprod{A'x}{s_x^{\star}} -g_{\mathrm{u}}\hspace{-2.5mm}\protoast \!(-{B'}^\top s_x^{\star}) - J^{*}(s_x^{\star})
=\min_{u \in \U^c} \big\{g_{\mathrm{u}}(u) + \inprod{A'x + B'u}{ s_x^{\star}}  \big \}  - J^{*}(s_x^{\star}) \, .
\end{align*}
Hence, by definition of the DP operator we have for all $x\in\cX_N^d$
\begin{align}
\eps 
&\geq |\mathrm{DP}[J](x)-\mathrm{DP_{conj}}[J'](x)| \nonumber \\
&=\left|\min_{u \in \U^c} \{ g_{\mathrm{u}}(u) + J (A'x+B'u) \} - h^{*}(A'x) \right|\nonumber\\
&=\left|\min_{u \in \U^c} \{ g_{\mathrm{u}}(u) + J (A'x+B'u) \} - \min_{u \in \U^c} \big\{g_{\mathrm{u}}(u) + \inprod{A'x + B'u}{ s_x^{\star}}  \big \}  + J^{*}(s_x^{\star}) \right| \, . \label{eq_step1_ds}
\end{align}
In addition we have for all $x\in \cX_N^d$ and $u\in \U^c$
\begin{align}
   g_{\mathrm{u}}(u)+ J(A'x+B'u)
   &=g_{\mathrm{u}}(u)+ J\!\protoast\!\protoast(A'x+B'u) \nonumber \\
   &\geq g_{\mathrm{u}}(u)+ J^{**}(A'x+B'u) - \eps \nonumber \\
   &=g_{\mathrm{u}}(u)+ \max_{z \in \cS^d_K}\{ \inprod{A'x+B'u}{z} - J^*(z) \}- \eps \nonumber \\
   &\geq g_{\mathrm{u}}(u)+\inprod{A'x+B'u}{s_x^{\star}} - J^*(s_x^{\star})- \eps \, , \label{eq_step2_ds}
\end{align}
where the second step uses Statement~\eqref{it_LFT3} of Lemma~\ref{lem_discreteLFT_approx} together with~\eqref{eq_byass}.
Combining~\eqref{eq_step1_ds} with~\eqref{eq_step2_ds} implies the assertion of the lemma, i.e., $ \norm{\pi^{\star}(x) - \hat \pi(x)} \leq \sqrt{4\eps/\mu_{g_{\mathrm{u}}}}$. 
To see this, recall that for any fixed $x\in\mathbb{X}^d$ we defined $\pi^{\star}(x)\in\mathbb{U}^d$ and $\hat \pi(x) \in\mathbb{U}^d$ as the minimizers of the first and second terms of~\eqref{eq_step1_ds}, respectively.
By definition of strong convexity we have
\begin{align*}
g_{\mathrm{u}}\big(\pi^\star(x)\big) 
&\geq g_{\mathrm{u}}\big(\hat \pi(x)\big) + \inprod{\pi^\star(x)-\hat\pi(x)}{\nabla g_{\mathrm{u}}\big(\hat \pi(x)\big) } + \frac{\mu_{g_{\mathrm{u}}}}{2}  \norm{\pi^{\star}(x) - \hat \pi(x)}^2 \\
&= g_{\mathrm{u}}\big(\hat \pi(x)\big) - \inprod{\pi^\star(x)-\hat\pi(x)}{B'^\top s^\star_x } + \frac{\mu_{g_{\mathrm{u}}}}{2}  \norm{\pi^{\star}(x) - \hat \pi(x)}^2 \\
&= g_{\mathrm{u}}\big(\hat \pi(x)\big) - \inprod{B'\pi^\star(x)-B'\hat\pi(x)}{s^\star_x } + \frac{\mu_{g_{\mathrm{u}}}}{2}  \norm{\pi^{\star}(x) - \hat \pi(x)}^2 \, ,
\end{align*}
where the pentultimate step uses the first order optimality condition of the optimization problem~\eqref{eq_optPolicy_via_LFT_robust} which ensures that $\nabla g_{\mathrm{u}}(\hat \pi(x))=-B'^\top s_x^\star$.
Hence, we find
\begin{align}
 &g_{\mathrm{u}}\big(\pi^\star(x)\big)+\inprod{A'x+B'\pi^\star(x)}{s_x^{\star}} - J^*(s_x^{\star}) -\eps \nonumber \\
 &\hspace{30mm}\geq  g_{\mathrm{u}}\big(\hat \pi(x)\big)+\inprod{A'x+B' \hat \pi(x)}{s_x^{\star}} - J^*(s_x^{\star}) + \frac{\mu_{g_{\mathrm{u}}}}{2}  \norm{\pi^{\star}(x) - \hat \pi(x)}^2 -\eps\nonumber \\
 &\hspace{30mm}\geq g_{\mathrm{u}}\big(\pi^\star(x)\big) + J\big(A'x+B'\pi^\star(x)\big) + \frac{\mu_{g_{\mathrm{u}}}}{2}  \norm{\pi^{\star}(x) - \hat \pi(x)}^2-2\eps \,,\label{eq:contradiction}
\end{align}
where the last step follows from~\eqref{eq_step1_ds}. Since we assumed in the beginning that $\|\pi^\star(x)-\hat \pi(x)\| > \sqrt{4\eps/\mu_{g_{\mathrm{u}}}}$, we have $\frac{\mu_{g_{\mathrm{u}}}}{2}  \norm{\pi^{\star}(x) - \hat \pi(x)}^2-2\eps>0$. Therefore \eqref{eq:contradiction} implies
\begin{align*}
    g_{\mathrm{u}}\big(\pi^\star(x)\big)+\inprod{A'x+B'\pi^\star(x)}{s_x^{\star}} - J^*(s_x^{\star})-\eps
    > g_{\mathrm{u}}\big(\pi^\star(x)\big) + J\big(A'x+B'\pi^\star(x)\big)\, ,
\end{align*}
which contradicts \eqref{eq_step2_ds} and hence proves the assertion of the lemma.\qed

\subsection{Proof of Lemma~\ref{lem_condNumber_valueFunction}} \label{app_proof_lemma_condNumber}
It is known~\cite{convexAnalysis93} that $\nabla f$ is $L'$-Lipschitz continuous if, and only if  $f^{\ast}$ is $1/L'$-strongly convex. Furthermore, by definition of the Lipschitz constant and the strong convexity parameter we have for $f=f_1+f_2$ that $L'_{f}\leq L'_{f_1}+L'_{f_2}$ and $\mu_{f}\geq \mu_{f_1}+\mu_{f_2}$.
Hence, we find
\begin{align*}
   L'_{J'_t}
   \leq L'_{h^*} + L'_{g_{\mathrm{x}}}
   = \frac{1}{\mu_{h}} + L'_{g_{\mathrm{x}}}
   \leq \frac{1}{\mu_{J^*_{t+1}} + \mu_{g_{\mathrm{u}}\hspace{-1.5mm}\protoast}}+ L'_{g_{\mathrm{x}}}
   = \frac{L'_{J_{t+1}} L'_{g_{\mathrm{u}}} }{L'_{J_{t+1}}+L'_{g_{\mathrm{u}}}}  + L'_{g_{\mathrm{x}}} \, . 
\end{align*}
Because the function $\R_+ \ni x \mapsto \alpha x/(\alpha + x) +\beta$ for $\alpha,\beta \in \R_+$ is monotonically increasing, this recursive inequality can be solved to obtain an explicit function $\varphi$ satisfying $L'_{J'_t} \leq \varphi(t,T,L'_{g_{\mathrm{x}}},L'_{g_{\mathrm{u}}},L'_{J_T})$ where 
\begin{align*} 
\varphi(t,T,x,y,z)
:= \frac{\sqrt{x(x+4y)}\, z\, \nu_1(t,T,x,y,z)+ x(2y+ z)\nu_2(t,T,x,y,z)}{\sqrt{x(x+4y)}\,z\, \nu_1(t,T,x,y,z) +(x- 2z)\nu_2(t,T,x,y,z) } \, , 
\end{align*}
for $\alpha_{\pm}(t,T,x,y,z):=-(x+2y \pm \sqrt{x(x+4y)})/z^2$ and
\begin{align*}
   \nu_1(t,T,x,y,z)&:=\alpha_{-}(t,T,x,y,z)^T\alpha_+(t,T,x,y,z)^t+\alpha_{-}(t,T,x,y,z)^t\alpha_+(t,T,x,y,z)^T \\
   \nu_2(t,T,x,y,z)&:=\alpha_{-}(t,T,x,y,z)^T\alpha_+(t,T,x,y,z)^t - \alpha_{-}(t,T,x,y,z)^t\alpha_+(t,T,x,y,z)^T\, .
\end{align*}
Analogously we obtain the recursive relation
\begin{align*}
   \mu_{J'_t} 
   \geq \mu_{h^*} + \mu_{g_{\mathrm{x}}}
   = \frac{1}{L'_{h}} + \mu_{g_{\mathrm{x}}}
   \geq \frac{1}{{L'}_{J^*_{t+1}} +L'_{g_{\mathrm{u}}\hspace{-1.5mm}\protoast}}+ \mu_{g_{\mathrm{x}}}
   = \frac{\mu_{J_{t+1}} \mu_{g_{\mathrm{u}}} }{\mu_{J_{t+1}}+\mu_{g_{\mathrm{u}}}}+ \mu_{g_{\mathrm{x}}} \, , 
\end{align*}
which again can be solved to obtain $ \mu_{J'_t} \geq \varphi(t,T,\mu_{g_{\mathrm{x}}},\mu_{g_{\mathrm{u}}},\mu_{T})$.
Combining this for
\begin{align}
  \phi(t,T,L'_{g_{\mathrm{x}}},\mu_{g_{\mathrm{x}}},L'_{g_{\mathrm{u}}},\mu_{g_{\mathrm{u}}},L'_{J_T},\mu_{J_T}):= \frac{\varphi(t,T,L'_{g_{\mathrm{x}}},L'_{g_{\mathrm{u}}},L'_{J_T})}{\varphi(t,T,\mu_{g_{\mathrm{x}}},\mu_{g_{\mathrm{u}}},\mu_{J_T})}\, , \label{eq_phi_def} 
\end{align}
gives
\begin{align*}
    \kappa_{J'_t}
    =\frac{L'_{J'_t}}{\mu_{J'_t}}
    \leq \frac{\varphi(t,T,L'_{g_{\mathrm{x}}},L'_{g_{\mathrm{u}}},L'_{J_T})}{\varphi(t,T,\mu_{g_{\mathrm{x}}},\mu_{g_{\mathrm{u}}},\mu_{J_T})}
    =\phi(t,T,L'_{g_{\mathrm{x}}},\mu_{g_{\mathrm{x}}},L'_{g_{\mathrm{u}}},\mu_{g_{\mathrm{u}}},L'_{J_T},\mu_{J_T})\, ,
\end{align*}
which proves the assertion.\qed

\subsection{Proof of Theorem~\ref{thm_QDP}} \label{app_pf_QDP}
We recall that $K=K_r K_i$, $N=N_{\dr} N_{\ci}$, and $d=d_{\dr}+ d_{\ci}$.
The regular dual spaces \smash{$\cS_{K_r}^{d_{\dr}}=\{s_0,\ldots,s_{K_r-1}\}$} and \smash{${\cS'}_{K_r}^{d_{\dr}}=\{s'_0,\ldots,s'_{K_i-1}\}$} are chosen such that $K \sim (T/\eps)^d$ which by~\eqref{eq_scaling_param} implies 
\begin{align} \label{eq_eps_disc}
    \eps_{\mathrm{disc}} \leq \frac{\eps}{T}\, ,
\end{align}
where we used that by assumption the DP operator preserves convexity.
The two QLFT steps that are required for each time step $t$ are probabilistic, and the condition numbers of the functions to which the QLFT is applied control the success probability. The iteration at stage $t=T,\ldots,1$ is successful with probability
\begin{align*}
    \frac{1}{\kappa_{J_{t}}^d \kappa_{h_{t}}^d}
    \geq \frac{1}{\kappa_{J_{t}}^{2d}}
    \geq \frac{1}{\phi(t,T,L'_{g_{\mathrm{x}}},\mu_{g_{\mathrm{x}}},L'_{g_{\mathrm{u}}},\mu_{g_{\mathrm{u}}},L'_{J_T},\mu_{J_T})^{2d}}\, ,
\end{align*}
where we the first inequality uses the facts that the condition number of a sum of two functions can be bounded from below by the sum of the individual condition numbers, and that the LFT does not change the condition number~\cite{QLFT20}. The final step follows from Lemma~\ref{lem_condNumber_valueFunction}. By definition of the function $\phi$ given in~\eqref{eq_phi_def} it follows that $\phi$ is monotonically decreasing in its first argument $t$.
Hence, because all $T$ runs need to be successful, the overall probability of success can be bounded by
\begin{align*}
   \frac{1}{\gamma^{2dT}}
   :=\frac{1}{\phi(0,T,L'_{g_{\mathrm{x}}},\mu_{g_{\mathrm{x}}},L'_{g_{\mathrm{u}}},\mu_{g_{\mathrm{u}}},L'_{J_T},\mu_{J_T})^{2dT}} \, .
\end{align*}
Combining our algorithm with amplitude amplification thus shows that we have to perform $\gamma^{dT}$ rounds, in expectation, to be successful with a constant probability.

The correctness of the algorithm follows from the correctness of the classical approach, which ensures that $h_\ell^*(q_i,r_{i'}) + g_{\mathrm{x}}(y_i,z_{i'}) = \hat{J}_{\ell-1}(y_i,z_{i'})$ for all $i \in [N_{\dr}],~i'\in [N_{\ci}]$. The nontrivial part is to ensure that all the steps above can be run efficiently on a quantum computer.

We go through the complexity of the different steps in Algorithm~\ref{algo_QDP}. The initialization takes $O(\polylog(N))$ time, since, by Assumption~\ref{ass_Uf}, we can load of $\hat{J}_T(\cdot)$ efficiently.
We next analyze the complexity of all the steps:
\begin{enumerate}
\item This calculation can be done in $O(\polylog(N,K))$ time because Proposition~\ref{prop_QLFT} ensures that the mapping
\begin{align}
&\frac{1}{\sqrt{N}} \sum_{i=0}^{N_{\dr}-1}\sum_{i'=0}^{N_{\ci}-1} \ket{i,i'}\ket{\hat{J}_{\ell}(y_{i-2\ell}^{i+2\ell},{z}_{i'-2\ell}^{i'+2\ell})}\ket{\textnormal{Garbage}(i,i')} \nonumber \\
&\hspace{20mm}\to  \frac{1}{\sqrt{K}} \sum_{j=0}^{K_r-1}\sum_{j'=0}^{K_i-1} \ket{j,j'} \ket{\hat{J}^*_{\ell}(s_{j-2\ell+1}^{j+2\ell-1},{s'}_{j'-2\ell+1}^{j'+2\ell-1})} \ket{\textnormal{Garbage}(j,j')} \, , \label{eq_step1}
\end{align}
if successful, requires $O(\polylog(N,K))$ steps. 
\item This computation can be done in $O(\polylog(K))$ time. To see this note that~\eqref{eq_step1} can be transformed into 
\begin{align*}
\frac{1}{\sqrt{K}} \sum_{j=0}^{K_r-1}\sum_{j'=0}^{K_i-1} \ket{j,j'} \ket{o_{j-2\ell+1}^{j+2\ell-1},p_{j-2\ell+1}^{j+2\ell-1}+{p'}_{j'-2\ell+1}^{j'+2\ell-1}} \ket{\hat{J}^*_{\ell}(s_{j-2\ell+1}^{j+2\ell-1},{s'}_{j'-2\ell+1}^{j'+2\ell-1})}  \ket{\textnormal{Garbage}(j,j')} \, ,
\end{align*}
as we know $(o_0,\ldots,o_{K_r-1})$, $(p_0,\ldots,p_{K_r-1})$, and $(p'_0,\ldots,p'_{K_i-1})$ we can perform $\ket{j,j'}\ket{0,0} \mapsto \ket{j,j'}\ket{o_{j-2\ell+1}^{j+2\ell-1},p_{j-2\ell+1}^{j+2\ell-1}+{p'}_{j'-2\ell+1}^{j'+2\ell-1}}$.
Since we know $g_{\mathrm{u}}\hspace{-1.5mm}\protoast$ we can create
\begin{align*}
\frac{1}{\sqrt{K}} \sum_{j=0}^{K_r-1}\sum_{j'=0}^{K_i-1} \ket{j,j'} \ket{g^*_{\mathrm{u}}(o_{j-2\ell+1}^{j+2\ell-1},p_{j-2\ell+1}^{j+2\ell-1}+{p'}_{j'-2\ell+1}^{j'+2\ell-1})}  \ket{\hat{J}^*_{\ell}(s_{j-2\ell+1}^{j+2\ell-1},{s'}_{j'-2\ell+1}^{j'+2\ell-1})}  \ket{\textnormal{Garbage}(j,j')}
 \, ,
\end{align*}
because by Assumption~\ref{ass_Uf} we can do $\ket{j,j'} \ket{o_j,p_j+{p'}_{j'}}\ket{0} \to \ket{j,j'}\ket{o_j,p_j+{p'}_{j'}} \ket{g_{\mathrm{u}}\hspace{-1.5mm}\protoast \!(o_j,p_j+{p'}_{j'})}$ efficiently. Using quantum arithmetics to add the content of two registers we obtain
\begin{align*}
&\frac{1}{\sqrt{K}} \sum_{j=0}^{K_r-1}\sum_{j'=0}^{K_i-1} \ket{j,j'}  \ket{\hat{J}^*_{\ell}(s_{j-2\ell+1}^{j+2\ell-1},{s'}_{j'-2\ell+1}^{j'+2\ell-1})+g^*_{\mathrm{u}}(o_{j-2\ell+1}^{j+2\ell-1},p_{j-2\ell+1}^{j+2\ell-1}+{p'}_{j'-2\ell+1}^{j'+2\ell-1})}  \ket{\textnormal{Garbage}(j,j')} \\
&\hspace{20mm}=\frac{1}{\sqrt{K}} \sum_{j=0}^{K_r-1}\sum_{j'=0}^{K_i-1} \ket{j,j'}\ket{h_{\ell}(s_{j-2\ell+1}^{j+2\ell-1},{s'}_{j'-2\ell+1}^{j'+2\ell-1})} \ket{\textnormal{Garbage}(j,j')} \, ,
\end{align*}
as desired.

\item This calculation can be done, if successful, in $O(\polylog(N,K))$ time by using the regular QLFT described in~\eqref{eq_iterative_regular_QLFT}.

\item This computation can be done in $O(\polylog(N))$ time. The state from the previous step can be transformed to
\begin{align}
\frac{1}{\sqrt{N}} \sum_{i=0}^{N_{\dr}-1} \sum_{i'=0}^{N_{\ci}-1} \ket{i,i'} \ket{h^*_{\ell}(q_{i-2(\ell-1)}^{i+2(\ell-1)},r_{i'-2(\ell-1)}^{i'+2(\ell-1)})} \ket{y_{i-2(\ell-1)}^{i+2(\ell-1)},z_{i'-2(\ell-1)}^{i'+2(\ell-1)}} \ket{\textnormal{Garbage}(i,i')} \, , \label{eq_almost_done}
\end{align}
since we know the vectors $(y_0,\ldots,y_{N_{\dr}-1})$, $(z_0,\ldots,z_{N_{\ci}-1})$, thus we can perform $\ket{i,i'}\ket{0,0} \mapsto \ket{i,j}\ket{y_{i-2(\ell-1)}^{i+2(\ell-1)},z_{i'-2(\ell-1)}^{i'+2(\ell-1)}}$.
State~\eqref{eq_almost_done} can then be turned into
\begin{align*}
\frac{1}{\sqrt{N}}\sum_{i=0}^{N_{\dr}-1} \sum_{i'=0}^{N_{\ci}-1} \ket{i,i'} \ket{h^*_{\ell}(q_{i-2(\ell-1)}^{i+2(\ell-1)},r_{i'-2(\ell-1)}^{i'+2(\ell-1)})}\ket{g_{\mathrm{x}}(y_{i-2(\ell-1)}^{i+2(\ell-1)},z_{i'-2(\ell-1)}^{i'+2(\ell-1)})} \ket{\textnormal{Garbage}(i,i')} \, ,
\end{align*}
where we used Assumption~\ref{ass_Uf} ensuring that the mapping $\ket{i} \ket{x_i}\ket{0} \to \ket{i} \ket{x_i} \ket{g_{\mathrm{x}}(x_i)}$ can be done efficiently. Adding the two registers then completes this step.
\end{enumerate}
The triangle inequality, together with the fact that the DP operator is contractive, implies that for all $x \in \cX^d_N$
\begin{align*}
&|\mathrm{DP}\circ \mathrm{DP}[J_t](x) - \mathrm{DP_{\!conj}}\circ\mathrm{DP_{\!conj}}[J'_t](x)|\\
&\hspace{10mm}\leq |\mathrm{DP}\circ \mathrm{DP}[J_t](x) - \mathrm{DP_{\!conj}}\circ\mathrm{DP}[J_t](x)| + |\mathrm{DP_{\!conj}}\circ \mathrm{DP}[J_t](x) - \mathrm{DP_{\!conj}}\circ\mathrm{DP_{\!conj}}[J'_t](x)|\\
&\hspace{10mm}\leq |\mathrm{DP}[J_{t-1}](x) - \mathrm{DP_{\!conj}}[J_{t-1}](x)| +|\mathrm{DP}[J_t](x) -\mathrm{DP_{\!conj}}[J'_t](x)| \\
&\hspace{10mm}\leq 2\eps_\mathrm{disc} \, ,
\end{align*}
where the final step uses~\eqref{eq_errorBound} from Theorem~\ref{thm_classical_error} and the fact that $J_t(x)=J'_t(x)$ for all $x \in \cX^d_N$. Applying this argument $T$ times, in an inductive fashion, together with~\eqref{eq_eps_disc} proves the assertion.\qed
\subsection{Proof of Corollary~\ref{cor_QDP_stoch}} \label{app_pf_cor_QDP_stoch}
We start by recalling that the error bound from Corollary~\ref{cor_stoch_error} ensures that for $K \sim (T/\eps)^d$ we have $\eps_{\mathrm{disc}} \leq \eps/T$.
The two QLFT transforms (each on $r$ registers) at each time step $t$ are successful with probability
\begin{align*}
    \frac{1}{\kappa_{V_{t}}^{dr} \kappa_{h_{t}}^{dr}}
    \geq \frac{1}{\kappa_{V_{t}}^{2dr}}
    \geq \frac{1}{\phi(t,T,L'_{g_{\mathrm{x}}},\mu_{g_{\mathrm{x}}},L'_{g_{\mathrm{u}}},\mu_{g_{\mathrm{u}}},L'_{ V_T},\mu_{V_T})^{2dr}}\, ,
\end{align*}
where the last inequality uses the fact that the condition number of $h^*(A'(m+\xi_k))+g_{\mathrm{x}}(m+\xi_k)$ is the same as the condition number of $\sum_{k=0}^{r-1} p_{\xi}(\xi_k) (h^*(A'(m+\xi_k))+g_{\mathrm{x}}(m+\xi_k))$. Following the same reasoning as in the proof of Theorem~\ref{thm_QDP} yields
\begin{align*}
   \frac{1}{\gamma^{2drT}}
   :=\frac{1}{\phi(0,T,L'_{g_{\mathrm{x}}},\mu_{g_{\mathrm{x}}},L'_{g_{\mathrm{u}}},\mu_{g_{\mathrm{u}}},L'_{V_T},\mu_{V_T})^{2drT}}
   \geq \frac{1}{(\kappa_{V_T}+T\kappa_{g_{\mathrm{u}}}+T\kappa_{g_{\mathrm{x}}})^{2drT}} \, ,
\end{align*}
where the final inequality uses Lemma~\ref{lem_condNumber_valueFunction}.
Combining our algorithm with amplitude amplification thus shows that we have to perform on average $\gamma^{drT}$ amplification rounds to be successful with a constant probability.

It thus remains to very that all steps in Algorithm~\ref{algo_QDP_stoch} can be done efficiently. Steps~\ref{it_stoch1}-\ref{it_stoch2} are standard and similar to Algorithm~\ref{algo_QDP}, with the only difference that we perform the QLFT $r$ times in all the $r$ registers. (Several QLFT, with different dual spaces, can be applied on the same input register because we never modify or overwrite the input registers.) Hence, the running time for these two steps is $O(r\, \polylog(P,K))$. 
In Step~\ref{it_stoch3} we perform the QLFT $r$ times with dual space $o_{k,i}$, where $k=0,\ldots,r-1$ is different in all $r$ registers. We then use quantum arithmetics to add the known function $g_{\mathrm{x}}$ evaluated at $\bar o_{k,i}$. The time complexity for this step is $O(r\, \polylog(P,K))$. 
To see how Step~\ref{it_stoch4} can be performed,  recall that for all $i \in [P]$
\begin{align*}
\hat{V}_{\ell-1}(m_{i})
= \sum_{k=0}^{r-1} p_{\xi}(\xi_k) \Big(h^*\big(A'(m_i+\xi_k)\big)+g_{\mathrm{x}}(m_i+\xi_k) \Big) 
= \sum_{k=0}^{r-1} p_{\xi}(\xi_k) \big(h^{*}(o_{k,i}) + g_{\mathrm{x}}(\bar o_{k,i}) \big)\, .
\end{align*}
We thus see that to transform the state from Step~\ref{it_stoch3} to the state from Step~\ref{it_stoch4} all we need to do is to compute the expectation with respect to the known discrete distribution $p_{\xi}(\xi_k)$, using the values contained in $r$ separate registers. This can be done with standard quantum arithmetics. Computing the value function and the optimal policy at a specific point then follows using similar steps to Corollary~\ref{cor_eval_J0} and Corollary~\ref{cor_optimal_policy_stoch}. We remark that in each QLFT step we accumulate a small amount of garbage, for a total of $O(rT \polylog(P,K)$ qubits. These extra qubits do not effect the probability of success of the algorithm because the postselection process in the QLFT algorithm~\cite{QLFT20} does not depend on them.\qed
\subsection{Proof of Proposition~\ref{prop_hardDP}}
To show this, we find a reduction from the problem described next, to the problem of calculating the optimal initial action of~\eqref{eq_hardDPproblem}.
\begin{problem}[Evaluating the CDF of the convolution of discrete random variables] \label{problem_CDF} \mbox{}\\
Instance: Discrete random variables $Z_1,\dots,Z_n$, $\Lambda \in \N$ and $\lambda \in \Q$ with $0 < \lambda \leq 1$. \\ 
Question: Is $\PP(\sum_{i=1}^n Z_i \le \Lambda) \geq \lambda$?
\end{problem}

\begin{proposition}[{\cite[Theorem~4.1]{halman09}}]
Problem~\ref{problem_CDF} is \#P-hard, even if the random variables $Z_i$ are independent and they have support $\{0, a_i\}$ with probability $\frac{1}{2}$ each.
\end{proposition}

\begin{proof}[Proof of Proposition~\ref{prop_hardDP}] \label{app_pf_hardness}
  The proof technique is inspired by \cite{halman2018complexity}. 
  Let $m = \max_{i} \alpha_{i,1}$, $g_T(x) = mn x^2$ for $m\geq 0$ sufficiently large, $g_{\mathrm{u},0}(u) = \beta u^2 + (1-\lambda) u$, $g_{\mathrm{u},t}(u) =\beta u^2+ u \; \forall t = 2,\dots,T-2$, and $g_{\mathrm{u},T-1}(u) = \beta u^2$, where $\beta>0$ will be determined subsequently and $\lambda>0$. Let $U_{\mathrm{x}}\geq \max\{1/\beta,mn\}$ and $U_{\mathrm{u}}\geq 1/\beta$. Note that $\xi_1,\ldots,\xi_{T}$ are discrete independent random variables with finite support $r>0$. 
  
  We first analyze the problem assuming $\beta = 0$, as it is easier to understand. In this case, the cost functions are not quadratic strongly convex; we will later analyze the problem with $\beta > 0$. Notice that the only way to increase the value of the state is via the action $u$, because of the structure of the transition function. Furthermore, at stages $t=2,\dots,T-1$ the state decreases by a random amount equal to $Z_{t-1}$. Finally, at stage $T-1$ we can decrease the value of the state ``for free'', as $b_{T-1} = -1$ and $g_{\mathrm{u},T-1}(u) = 0$.

  Since the terminal cost function $g_{T}(\cdot)$ is a quadratic with a cost coefficient that is larger than the cost for buying even a single unit, it is optimal to reach state $x_T = 0$ at the last stage, since the penalty for $x_T \neq 0$ is larger than any other cost. Also notice that, because at stage $T-1$ we can decrease the value of the state variable ``for free'', we can assume that any optimal solution must reach stage $T-1$ with a state variable $x_{T-1} \ge 0$. 

  At the first stage of the dynamic program we have the option of buying at price $(1-\lambda)$ per unit. In the following stages we have to satisfy demand $Z_t$ using stored resource, or buying (via the action $u_t$) at cost $1$. Note that since costs do not vary after the first stage, it is always optimal to use stored resource as much as possible, then buy via the action $u_t$ the exact amount necessary to ensure that we reach $x_{T-1} \ge 0$ in case the state becomes negative. Hence, there is only one decision to take: the amount $u_0^\star$ of energy that is bought at stage $0$ with unit cost $(1-\lambda)$, and the remaining amount $(\sum_{i=1}^n Z_i - u_0^\star)^+$ must be bought in subsequent stages at unit cost $1$. It follows that this problem is equivalent to a newsvendor problem with demand $\sum_{i=1}^n Z_i$ where the unit cost for overbuying is $(1-\lambda)$ and the unit cost for underbuying is $1 - (1-\lambda) = \lambda$. It is well known that in this case, the optimal amount $u_0^\star$ to be bought is:
  \begin{align*}
    u_0^\star = \arg \min_{z\in \R^d} \left\{ z: \PP\left(\sum_{i=1}^n Z_i \leq z \right) \geq
    \frac{\lambda}{\lambda + (1 - \lambda)} = \lambda \right\}.
  \end{align*}
  Hence, if we could determine the optimal policy for this DP with stochastic demand in polynomial time, we would be able to solve Problem~\ref{problem_CDF} in polynomial time. This implies that determining the optimal policy $u_0^\star$ for this form of DP, where we assumed $\beta = 0$, is \#P-hard.
  
  We finally argue that the cost functions can be made quadratic strongly convex while yielding almost the same optimal action. Let $u_0^\star$ be the optimal action at the first stage (i.e., order quantity in the equivalent newsvendor) when $\beta = 0$, and $J_0(0)$ the corresponding optimal objective function value of the DP at initial state $x_0=0$. Define $\bar{\lambda} := \min \{\lambda, 1-\lambda\}$. Notice that $u_0^\star$ is integer, because all the $X_i$ are integer. Now consider a new instance of the problem where we set $\beta = \bar{\lambda}/(8 m^2 n^2)$. As compared to the instance with linear costs, we need to pay additional quadratic costs at every stage. If we buy $u_0^\star$ at the first stage, we have to buy at most $mn-u_0^\star$ at subsequent stages. Thus, the cost of the action $u^\star$ with the quadratic cost function is at most $J_0(0) + \beta (u^\star)^2 + \beta (mn-u_0^\star)^2 \leq J_0(0) + \frac{\bar{\lambda}}{8} + \frac{\bar{\lambda}}{8} < J_0(0) + \frac{\bar{\lambda}}{2}$, where we used the fact that $u_0^\star \leq m n$. We also know that in the instance with linear cost, buying $\geq u_0^\star + 1$ units in the first stage has cost at least $J_0(0) + (1-\lambda)$, and buying $\le u_0^\star - 1$ units has cost at least $J_0(0) + \lambda$. By definition of $\bar{\lambda}$, this implies that it is suboptimal to buy an amount $\ge u_0^\star+1$ or $\le u_0^\star-1$ in the first stage. Furthermore, it is easy to see that buying a fractional amount $u' \in (u_0^\star, u_0^\star+1)$ can only increase the expected cost. If we buy a fractional amount $u' \in (u_0^\star-1, u_0^\star)$, say, $u_0^\star - \Delta_u$ with $0 < \Delta_u < 1$, the first-stage cost decreases by at most $(1-\lambda) \Delta_u + \frac{\bar{\lambda}}{8}$, but we still need to pay at least $\Delta_u$ (in expectation) in subsequent time stages to buy the difference. This means that in order for the cost to decrease, we must have:
  \begin{equation*}
      \Delta_u - (1-\lambda)\Delta_u - \frac{\bar{\lambda}}{8} < 0 \, ,
  \end{equation*}
  which implies $\Delta_u < \frac{\bar{\lambda}}{8\lambda} \leq \frac{1}{8}$. Hence, we have shown that the optimal first-stage action belongs to the interval $[u_0^\star-\frac{1}{8}, u_0^\star]$. It follows that simple rounding of the optimal first-stage action of the DP with quadratic cost is equal to $u_0^\star$; as shown above, determining $u_0^\star$ is \#P-hard, concluding the proof.
 \end{proof}
 
\bibliographystyle{arxiv_no_month}
\bibliography{bibliofile}

\end{document}